\documentclass[11pt]{article}
\usepackage{graphicx}
\usepackage{amsmath}
\usepackage{epsfig}
\usepackage{epsf}
\usepackage{amsfonts}
\usepackage{latexsym}
\usepackage{fancyvrb}
\usepackage{amssymb}
\usepackage{amsthm}
\usepackage{graphics}
\usepackage{psfrag}
\usepackage{float}
\usepackage{enumerate}
\usepackage{url}
\usepackage{fullpage}
\usepackage[small]{caption}

\newtheorem{theorem}{Theorem}
\newtheorem{lemma}{Lemma}

\newtheorem{corollary}{Corollary}

\newcommand{\later}[1]{{}}
\newcommand{\old}[1]{{}}

\newcommand{\eps}{\varepsilon}

\def\etal{{\it et~al.}\,}
\def\ie{{\it i.\,e.},~}
\def\eg{{\it e.\,g.},~}

\newcommand{\conv}{{\rm conv}}

\newcommand{\supp}{{\rm supp}}

\newcommand{\ar}{\rightarrow}

\newcommand{\squishlist}{
 \begin{list}{$\bullet$}
  { \setlength{\itemsep}{0pt}
     \setlength{\parsep}{3pt}
     \setlength{\topsep}{3pt}
     \setlength{\partopsep}{0pt}
     \setlength{\leftmargin}{1.5em}
     \setlength{\labelwidth}{1em}
     \setlength{\labelsep}{0.5em} } }

\newcommand{\squishend}{
  \end{list} }

\def\H{\mathcal H}

\def\tr{{\tt tr}}
\def\sc{{\tt sc}}
\def\st{{\tt st}}

\def\cf{{\tt cf}}
\def\pm{{\tt pm}}

\def\pg{{\tt pg}}
\def\ns{\hspace*{-1ex}}

\providecommand{\intd}[0]%
{\;\mbox{d}}

\begin{document}

\title{Bounds on the maximum multiplicity\\
 of some common geometric graphs}

\author{
Adrian Dumitrescu\thanks{%
Department of Computer Science,
University of Wisconsin--Milwaukee, USA\@.
Supported in part by NSF grants CCF-0444188 and DMS-1001667.
Email:~\texttt{dumitres@uwm.edu}}
\and
Andr\'e Schulz\thanks{%
Institut f\"ur Mathematische Logik und Grundlagenforschung, Universit\"at M\"unster,
Germany\@. Partially supported by the German Research Foundation (DFG)
under grant SCHU 2458/1-1.
Email: \texttt{andre.schulz@uni-muenster.de} }
\and
Adam Sheffer\thanks{%
School of Computer Science, Tel Aviv University, Tel Aviv 69978, Israel\@.
Supported by Grant 338/09 from
the Israel Science Fund.
Email: \texttt{sheffera@tau.ac.il} }
\and
Csaba D. T\'oth\thanks{%
Department of Mathematics and Statistics,
University of Calgary, Canada\@. Supported in part by NSERC grant
RGPIN 35586.
Research by this author was conducted at Tufts University and at
Ecole Polytechnique F\'ed\'erale de Lausanne.
Email: \texttt{cdtoth@ucalgary.ca} }
}

\maketitle


\begin{abstract}
We obtain new lower and upper bounds for the maximum multiplicity of
some weighted and, respectively, non-weighted common geometric graphs
drawn on $n$ points in the plane in general position (with no three
points collinear): perfect matchings, spanning trees,
spanning cycles (tours), and triangulations.

(i) We present a new lower bound construction for the maximum number of
triangulations a set of $n$ points in general position can have. In
particular, we show that a {\em generalized double chain} formed by two
{\em almost convex} chains admits $\Omega (8.65^n)$ different
triangulations. This improves the bound $\Omega (8.48^n)$ achieved by
the previous best construction, the \emph{double zig-zag chain}
studied by Aichholzer~\etal

(ii) We obtain a new lower bound of $\Omega(12.00^n)$ for the number of
{\em non-crossing} spanning trees of the {\em double chain} composed of two
{\em convex} chains. The previous bound, $\Omega(10.42^n)$, stood
unchanged for more than 10 years.

(iii) Using a recent upper bound of $30^n$ for the number of
triangulations, due to Sharir and Sheffer, we show that $n$ points in the plane
in general position admit at most $O(68.62^n)$ non-crossing spanning cycles.

(iv) We derive lower bounds for the number of maximum and minimum
{\em weighted} geometric graphs (matchings, spanning trees, and tours).
We show that the number of shortest tours can be exponential in $n$ for
points in general position. These tours are automatically non-crossing.
Likewise, we show that the number of longest non-crossing
tours can be exponential in $n$. It was known that the number of shortest
non-crossing perfect matchings can be exponential in $n$, and here we show
that the number of longest non-crossing perfect matchings can be also
exponential in $n$. It was known that the number of longest non-crossing
spanning trees of a point set can be exponentially large, and here we
show that this can be also realized with points in convex position.
For points in convex position we obtain tight bounds for the
number of longest and shortest tours.
We also give a combinatorial characterization of longest tours,
which yields an $O(n\log n)$ time algorithm for computing them.
\end{abstract}

\medskip
\hspace{0.1in}
\textbf{\small Keywords:}
{\small Geometric graph, non-crossing property, Hamiltonian cycle, perfect
matching, triangulation, spanning tree.}

\section{Introduction}

Let $P$ be a set of $n$ points in the plane in {\em general position},
\ie no three points lie on a common line.
A {\em geometric graph} $G=(P,E)$ is a graph drawn in the plane so
that the vertex set consists of the points in $P$ and the edges
are drawn as straight line segments between the corresponding points in $P$.
All graphs we consider in this paper are geometric graphs. We call a
 graph \emph{non-crossing} if its edges intersect only at common endpoints.

Determining the maximum number of non-crossing geometric graphs on $n$
points in the plane is a fundamental question in combinatorial geometry.
We follow common conventions (see for instance~\cite{SW06})
and denote by $\pg(P)$ the number of non-crossing
geometric graphs that can be embedded over the planar point set $P$, and by $\pg(n)=\max_{|P|=n} \pg(P)$ the {\em maximum number}
of non-crossing graphs an $n$-element point set can admit.
Analogously, we introduce shorthand notation for the maximum number of
triangulations, perfect matchings, spanning trees, and spanning cycles
(\ie Hamiltonian cycles); see Table~\ref{table1}.
For example, $\tr(n)= O(f(n))$ means that any $n$-element point set
can have at most $O(f(n))$ triangulations, and $\tr(n)=\Omega(g(n))$
means that there exists some $n$-element point set that admits $\Omega(g(n))$
triangulations.

\begin{table}[h]
\begin{center}
\begin{tabular}{|c|c|c|c|}
\hline
{\textbf {Abbr.}} & {\textbf {Graph class}} & {\textbf {Lower bound}}
& {\textbf {Upper bound}} \\
\hline\hline
\pg(n) & graphs &
$\Omega(41.18^n)$~\cite{AHV+06,GNT00} & $O(207.84^n)$~\cite{HSSTW11,SS10}\\
\hline
\cf(n) & cycle-free graphs & $\mathbf{\Omega(12.26^n)}$~[new, Theorem~\ref{thm:lbst}] & $O(164.49^n)$~\cite{HSSTW11,SS10}\\
\hline
\pm(n) & perfect matchings & $\Omega^*(3^n)$~\cite{GNT00} & $O(10.07^n)$~\cite{SW06} \\
\hline
\st(n) & spanning trees & $\mathbf{\Omega(12.00^n)}$~[new, Theorem~\ref{thm:lbst}] & $O(146.37^n)$~\cite{HSSTW11,SS10} \\
\hline	
\sc(n) & spanning cycles & $\Omega(4.64^n)$~\cite{GNT00} & $\mathbf{O(68.62^n)}$~[new, Theorem~\ref{th:triVSsc}]\\
\hline
\tr(n) & triangulations & $\mathbf{\Omega(8.65^n)}$ [new, Theorem~\ref{T1}]& $O(30^n)$~\cite{SS10} \\
\hline
\end{tabular}
\caption{Classes of non-crossing geometric (straight line) graphs, current best
  upper and lower bounds.\label{table1}}
\end{center}
\vspace{-\baselineskip}
\end{table}

In the past 30 years numerous researchers have tried to estimate these
quantities. In a pivotal result, Ajtai~\etal\cite{ACNS82} showed
that $\pg(n)=O(c^n)$ for an absolute, but very large constant $c>0$.
The constant $c$ has been improved several times since then. The best
bound today is $c<207.84$, which follows from a combination of a
result of Sharir and Sheffer~\cite{SS10} with a result of
Hoffmann~\etal\cite{HSSTW11}. Interestingly, this upper bound, as well as
the currently best upper bounds for $\st(n)$, $\sc(n)$, and $\cf(n)$,
are derived from an upper bound on the maximum number of triangulations, $\tr(n)$.
This underlines the importance of the bound for $\tr(n)$ in this setting.
For example, the best known upper bound for $\st(n)$ is the combination of
$\tr(n)\leq 30^n$~\cite{SS10} with the ratio
$\sc(n)/\tr(n)=O^*\left(4.879^n\right)$~\cite{HSSTW11};
see also previous work~\cite{RSW08,SS03,SW06,SW06b}.
To our knowledge, the only upper bound derived via a different
approach is that for the number of perfect matchings by
Sharir and Welzl~\cite{SW06}, $\pm(n)=O(10.07^n)$.

So far, we recalled various upper bounds on the maximum number of
geometric graphs in certain classes. In this paper we mostly conduct
our offensive from the other direction, on improving the corresponding
\emph{lower bounds}. Lower bounds for unweighted non-crossing graph
classes were obtained in~\cite{AHV+06,D99,GNT00}.
Garc\'{\i}a, Noy, and Tejel~\cite{GNT00} were the first to recognize
the power of the {\em double chain} configuration in establishing
good lower bounds for the number of matchings, triangulations, spanning cycles and
trees. It was widely believed for some time that the double chain
gives asymptotically the highest number of triangulations, namely $\Theta^*(8^n)$.
This was until 2006, when Aichholzer~\etal\cite{AHV+06} showed that another
configuration, the so-called {\em double zig-zag chain},
admits $\Theta^*(\sqrt{72}^n)=\Omega(8.48^n)$
triangulations\footnote{The $\Theta^*,O^*, \Omega^*$ notation is used to
describe the asymptotic growth of functions ignoring polynomial
factors.}.
The double zig-zag chain consists of two flat copies of a zig-zag chain.
A zig-zag chain is the simplest example of an {\em almost convex}
polygon. Such polygons have been introduced and first studied by
Hurtado and Noy~\cite{HN97}. In this paper we further exploit
the power of {\em almost convex} polygons and establish a new lower bound
$\tr(n)=\Omega(8.65^n)$. For matchings, spanning cycles, and plane graphs
the double chain still holds the current record.

Less studied are multiplicities of {\em weighted} geometric graphs.
The weight of a geometric graph is the sum of its (Euclidean) edge lengths.
This leads to the question: how many graphs of a certain type
(\eg matchings, spanning trees, or tours)
with {\em minimum} or {\em maximum} weight can be realized on an $n$-element
point set. The notation is analogous; see Table~2. 
Dumitrescu~\cite{D02} showed that the longest and shortest
matchings can have exponential multiplicity, $2^{\Omega(n)}$, for a point
set in general position. Furthermore, the longest and shortest spanning trees
can also have multiplicity of $2^{\Omega(n)}$. Both bounds count
explicitly geometric graphs with crossings; however these minima are
automatically non-crossing. The question for the maximum multiplicity
for non-crossing geometric graphs remained open for most classes
of geometric graphs. Since we lack any upper bounds that are
better than those for the corresponding unweighted classes,
the ``Upper bound'' column is missing from Table~\ref{table2}.

\begin{table}[h]
\begin{center}
\begin{tabular}{|c|c|c|}
\hline
{\textbf {Abbr.}} & {\textbf {Graph class}} & {\textbf {Lower bound}}\\
\hline\hline
$\pm_{\rm min}(n)$ &shortest perfect matchings & $\Omega(2^{n/4})$~\cite{D02} \\
\hline
$\pm_{\rm max}(n)$ & longest perfect matchings & $\mathbf{\Omega(2^{n/4})}$~[new, Theorem~\ref{T2}] \\
\hline
$\st_{\rm min}(n)$ & shortest spanning trees & $\Omega(2^{n/2})$~\cite{D02} \\
\hline	
$\st_{\rm max}(n)$ & longest spanning trees & $\mathbf{\Omega(2^n)}$~[new, Theorem~\ref{T3}] \\
\hline	
$\sc_{\rm min}(n)$ & shortest spanning cycles & $\mathbf{\Omega(2^{n/3})}$~[new, Theorem~\ref{thm:min-tours}]\\
\hline
$\sc_{\rm max}(n)$ & longest spanning cycles & $\mathbf{\Omega(2^{n/3})}$~[new, Theorem~\ref{thm:max-tours}] \\
\hline
\end{tabular}
\caption{Classes of {\em weighted} non-crossing geometric
  graphs: exponential lower bounds.\label{table2}}
\end{center}
\vspace{-2\baselineskip}
\end{table}


\paragraph{Our results.} 

\squishlist
\item[(I)] A new lower bound, $\Omega(8.65^n)$, for the maximum number of
triangulations a set of $n$ points can have.
We first re-derive the bound given by Aichholzer~\etal~\cite{AHV+06}
with a simpler analysis, which allows us to extend it to
more complex point sets. Our estimate might be the best possible
for the type of construction we consider.

\item[(II)] A new lower bound, $\Omega(12.00^n)$, for the maximum number of
non-crossing spanning trees a set of $n$ points can have.
This is obtained by refining the analysis of the number of such trees
on the ``double chain'' point configuration. The previous bound was
$\Omega(10.42^n)$. Our analysis of the construction improves
also the lower bound for cycle-free non-crossing graphs
due to Aichholzer~\etal\cite{AHV+06}, from $\Omega(11{.}62^n)$ to $\Omega(12{.}23^n)$.

\item[(III)] A new upper bound, $O(68.62^n)$, for the number of
non-crossing spanning cycles on $n$ points in the plane. This
improves the latest upper bound of $O(70.22^n)$ that follows from a
combination of the results of Buchin~\etal\cite{BKK+07} and a
recent upper bound $\tr(n)\leq 30^n$ by Sharir and Sheffer~\cite{SS10}.

\item[(IV)] New bounds on the maximum multiplicity of various weighted
geometric graphs on $n$ points (weighted by Euclidean length).
We show that the maximum number of longest non-crossing perfect matchings,
spanning trees, spanning cycles, as well as shortest tours are all exponential in $n$.
We also derive tight bounds, as well as a combinatorial characterization of
longest tours (with crossings allowed) over $n$ points in convex position.
This yields an $O(n \log n)$ algorithm to compute a longest tour for such sets.
\squishend

The main results of this paper were presented at the \emph{28th
  International Symposium on Theoretical Aspects of Computer Science
  (STACS)} in March 2011, in Dortmund, Germany.
An extended abstract (without all the proofs) appeared in the corresponding
conference proceedings~\cite{DSST11}. Since then we were able to further
refine the analysis in two of our main theorems.
In particular, we improved the bounds in Theorems~\ref{thm:lbst}
and~\ref{thm:polygonize}.

\subsection{Preliminaries} \label{sec:prelim}

\paragraph{Asymptotics of multinomial coefficients.}
Denote by $H(q)=-q \log q - (1-q) \log (1-q)$ the {\em binary entropy
function}, where $\log$ stands for the logarithm in base 2.
(By convention, $0 \log{0} =0$.)
For a constant $0 \leq \alpha \leq 1$, the following estimate can be
easily derived from Stirling's formula for the factorial:
\begin{equation} \label{E1}
{n \choose {\alpha n}}= \Theta(n^{-{1/2}} 2^{H(\alpha)n}),
\end{equation}
We also need the following bound on the sum of binomial coefficients;
see~\cite{BST98} for a proof and~\cite{DK01,DS00} for an application.
If $0<\alpha \leq {1 \over 2}$ is a constant,
\begin{equation} \label{E21}
\sum_{k=0} ^{k \leq \alpha n} {n \choose k} \leq 2^{H(\alpha)n}.
\end{equation}

Define similarly the {\em generalized entropy function} of $k$
parameters $\alpha_1,\ldots,\alpha_k$, satisfying
\begin{equation} \label{E2}
\sum_{i=1}^k \alpha_i =1, \ \ \alpha_1,\ldots,\alpha_k \geq 0, \
\end{equation}

\begin{equation*} \label{E3}
H_k(\alpha_1,\ldots,\alpha_k)=-\sum_{i=1}^k \alpha_i \log{\alpha_i}.
\end{equation*}

Clearly, $H(q)=H_2(q,1-q)$. Recall, the multinomial coefficient
\begin{equation*} \label{E4}
{n \choose n_1,n_2,\ldots,n_k} = \frac{n!}{n_1! n_2! \ldots n_k!},
\end{equation*}
where $\sum_{i=1}^k n_i =n$, counts the number of distinct ways to
permute a multiset of $n$ elements, $k$ of which are distinct,
with $n_i$, $i=1,\ldots,k$,
being the multiplicities of each of the $k$ distinct elements.

Assuming that $n_i=\alpha_i n$, $i=1,\ldots,k$, for
constants $\alpha_1,\ldots,\alpha_k$, satisfying~\eqref{E2},
again by using Stirling's formula for the factorial, one gets
an expression analogous to~\eqref{E1}:
\begin{equation} \label{E5}
{n \choose n_1,n_2,\ldots,n_k} =
\Theta(n^{-(k-1)/2}) \cdot \left(\prod_{i=1}^k \alpha_i^{-\alpha_i}\right)^n=
\Theta(n^{-(k-1)/2}) \cdot 2^{H_k(\alpha_1,\ldots,\alpha_k)n}.
\end{equation}

\paragraph{Notations and conventions.}
For a polygonal chain $P$, let $|P|$ denote the number of vertices.
If $1<c_1<c_2$ are two constants, we frequently write
$\Omega^*(c_2^n) = \Omega(c_1^n)$.
A geometric graph $G=(V,E)$ is called a {\em (geometric) thrackle}, if any two
edges in $E$ either cross or share a common endpoint; see \eg\cite{PA95}.

\section {Lower bound on the maximum number of triangulations}
\label{sec:tri}

Hurtado and Noy~\cite{HN97} introduced the class $P(n,k^r)$ of
{\em almost convex polygons}. A simple polygon is in this class
if it has $n$ vertices, and it can be obtained from a convex $r$-gon
by replacing each edge with a ``flat'' reflex chain having $k$
interior vertices (thus there are $r$ reflex chains altogether), such
that any two vertices of the polygon can be connected by an internal
diagonal unless they belong to the same reflex chain.
For example, $P(n,0^r)$ is the class of convex polygons with $n=r$ vertices.
Note that, for $r \geq 3$ and $k \geq 0$, every polygon in $P(n,k^r)$ has
$n=r(k+1)$ vertices, $r$ of which are convex.

Following the notation from~\cite{HN97}, we denote by $t(n,k^r)$
the number of triangulations of a polygon $P(n,k^r)$. According to
\cite[Theorem 3]{HN97},
\begin{equation*} \label{E6}
t(n,k^r) =\Theta\left(\left(\frac{1+k/2}{2^k}\right)^r \cdot t(n,0^n)\right) =
\Theta\left(\left( \frac{k+2}{2^{k+1}}\right)^r \cdot 4^{r(k+1)}\right)=
\Theta\left(\left( (k+2)^{\frac{1}{k+1}} \cdot 2 \right)^n\right).
\end{equation*}
In particular,
\begin{itemize}
\item for $k=1$, $t(n,1^r) = \Theta( (2 \sqrt{3})^n)=\Theta(\sqrt{12}^n)$.
This estimate was used by Aichholzer~\etal~\cite{AHV+06} to show that
the double zig-zag chain has $\Omega(8.48^n)$ triangulations. 
\item for $k=2$, $t(n,2^r) = \Theta((2^{5/3})^n)$.
\item for $k=3$, $t(n,3^r) = \Theta((5^{1/4} \cdot 2)^n)$.
\item for $k=4$, $t(n,4^r) = \Theta((6^{1/5} \cdot 2)^n)$.
\end{itemize}

Analogously we define the class $P^+(n,k^r)$ of {\em almost convex (polygonal) chains}.
An $x$-monotone polygonal chain in this class has $n=r(k+1)+1$
vertices, and it can be obtained from an $x$-monotone convex chain
with $r+1$ vertices by replacing each edge with a ``flat'' reflex
chain with $k$ interior vertices, such that any two vertices of the
chain can be connected by a segment lying above the chain unless they are incident
to the same reflex chain. See Fig.~\ref{f6} for a small example.
To further simplify notation,
we denote by $P^+(n,k^r)$ any polygonal chain in this class; note they are all equivalent
in the sense that they have the same visibility graph. Note that the simple polygon
bounded by an almost convex chain $P^+(n,k^r)$ and the segment between its two extremal
vertices has $\Theta^*(t(n,k^r))$ triangulations.
\begin{figure}[htbp]
\centerline{\epsfysize=1.4in \epsffile{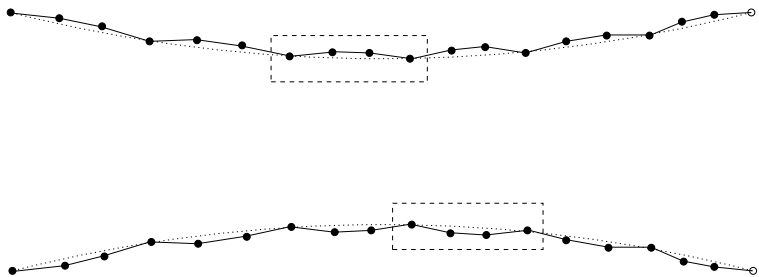}}
\caption{Two (flat) mutually visible copies of $P^+(19,2^6)$
with opposite orientations that form $D(38,2^{12})$. Two consecutive
hull vertices of $P^+(19,2^6)$ with a reflex chain of two vertices in
between are indicated in both the upper and the lower chain.}
\label{f6}
\end{figure}

In establishing our new lower bound for the maximum number of triangulations,
we go through the following steps: We first describe the double zig-zag chain
from~\cite{AHV+06} in our framework, and re-derive the $\Omega^*(\sqrt{72}^n)$
bound of~\cite{AHV+06} for the number of its triangulations. Our simpler analysis
extends to some variants of the double zig-zag chain, and leads to a new
lower bound on $\tr(n)=\Omega(8.65^n)$.

Two $x$-monotone polygonal chains $L$ and $U$ are said to be {\em mutually
visible} if every pair of points, $p \in L$ and $q \in U$, are {\em visible}
from each other (that is, the segment $pq$ crosses neither $U$ nor $L$).

Let us call $D(n,k^r)$ the {\em generalized double chain} of $n$ points formed by
the set of vertices in two mutually visible $x$-monotone chains, each
with $n/2=r(k+1)+1$ vertices, where the upper chain is a $P^+(n,k^r)$ and
the lower chain is a horizontally reflected copy of $P^+(n,k^r)$,
as in Fig.~\ref{f6}. Generalized double chains are a family of point configurations,
containing, among others, the double chain and double zig-zag chain configurations.
In particular, $D(n,1^r)$ is the \emph{double zig-zag chain} used
by Aichholzer~\etal\cite{AHV+06}.

\begin{theorem} \label{T1}
The point set $D(n,3^r)$ with $n=8r+2$ points admits $\Omega(8.65^n)$ triangulations.
Consequently, $\tr(n)=\Omega(8.65^n)$.
\end{theorem}
\begin{proof}
The following estimate is used in all our triangulation bounds.
Consider two mutually visible flat polygonal chains, $L$ and $U$, with $m$
vertices each ($L$ is the lower chain and $U$ is the upper chain). As
in the proof of~\cite[Theorem 4.1]{GNT00}, the region between the two
chains consists of $2m-2$ triangles, such that exactly $m-1$ triangles
have an edge along $L$ and the remaining $m-1$ triangles have an edge adjacent to $U$.
It follows that the number of distinct triangulations of this middle region is
\begin{equation} \label{E7}
{2m-2 \choose m-1} = \Theta(m^{-1/2} \cdot 4^m).
\end{equation}

\paragraph{The old $\Omega(8.48^n)$ lower bound in a new perspective.}
We estimate from below the number of triangulations of $D(n,1^r)$ as
follows. Recall that $|L|=|U|=n/2=2r+1$. Include all edges of $L$ and $U$
in any of the triangulations we construct.
Now construct different triangulations as follows. Independently
select a subset of $\alpha_1 r$ short edges of $\conv(U)$ and
similarly, a subset of $\alpha_1 r$ short edges of $\conv(L)$. Here
$\alpha_1 \in (0,1)$ is a constant to be chosen later.
According to~\eqref{E1}, this can be done in
$$ {r \choose \alpha_1 r} = \Theta(r^{-1/2} \cdot 2^{H(\alpha_1) r}) $$
ways in each of the two chains. Include these edges in the triangulation.
Observe that after adding these short edges the middle region between
the (initial) chains $L$ and $U$ is sandwiched between two
mutually visible shorter chains, say $L' \subset L$ and $U' \subset U$, where
\begin{equation} \label{E8}
|L'|=|U'|=2r-\alpha_1 r=(2-\alpha_1)r.
\end{equation}
Triangulate this middle regions in all of the possible ways,
as outlined in the paragraph above~\eqref{E7}. Let $N$
denote the total number of triangulations of $D(n,1^r)$ obtained in this way.
By the above estimate, we have $t(n,1^r) =\Theta((2 \sqrt{3})^n)$.
Combining this with~\eqref{E7} and~\eqref{E8},
\begin{align*} \label{E9}
N &= \Omega^* \left( \left[(2 \sqrt{3})^{2r} 2^{H(\alpha_1) r} \right]^2
4^{(2-\alpha_1)r} \right)=
\Omega^* \left(
\left[ 2^{2r} 3^r 2^{(2-\alpha_1)r} 2^{H(\alpha_1) r} \right]^2
\right) =
\nonumber \\
&= \Omega^* \left(
\left[ 2^{2} \cdot 3 \cdot 2^{(2-\alpha_1)} 2^{H(\alpha_1)} \right]^{2r}
\right) =
\Omega^* \left(
\left[ 2^{4-\alpha_1+H(\alpha_1)}\cdot 3 \right]^{n/2} \right) =
\Omega^* \left( a^n\right),
\end{align*}
where
$$ a=\left[ 2^{4-\alpha_1+H(\alpha_1)}\cdot 3 \right]^{(1/2)}.$$
By setting $\alpha_1=1/3$, as in~\cite{AHV+06}, this yields
$a= 6 \sqrt{2}= 8.485\ldots$, and $N = \Omega^*(8.485^n)=\Omega(8.48^n)$.

Applying a similar analysis for a generalized double chain with reflex chains of length
$3$ implies Theorem \ref{T1}.
To simplify the presentation we describe the next step
using almost convex polygonal chains with reflex chains of length $2$.
\paragraph{The next step: using $D(n,2^r)$.}
Notice that in this case $n=6r+2$.
In the upper chain $U$, each reflex chain of two points together with the two
hull vertices next to it form a convex $4$-chain, as viewed from below.
In each such $4$-chain, independently, we proceed with one of the
following three choices: (0) we leave it unchanged; (1) we add
one edge, so that the chain length is reduced by $1$ (from $4$ to
$3$); (2) we add two edges, so that the chain length is reduced by $2$
(from $4$ to $2$). We refer to these changes and corresponding length
reductions as reductions of type 0, 1 and 2, respectively.
For $i=0,1,2$, let $a_i$ count the number of distinct ways such
reductions can be performed on a convex $4$-chain of $U$.
See Fig.~\ref{f7}(left). Clearly, we have $a_0=1$ and $a_1=a_2=2$.
\begin{figure}[h]
\centerline{\epsfysize=1.22in \epsffile{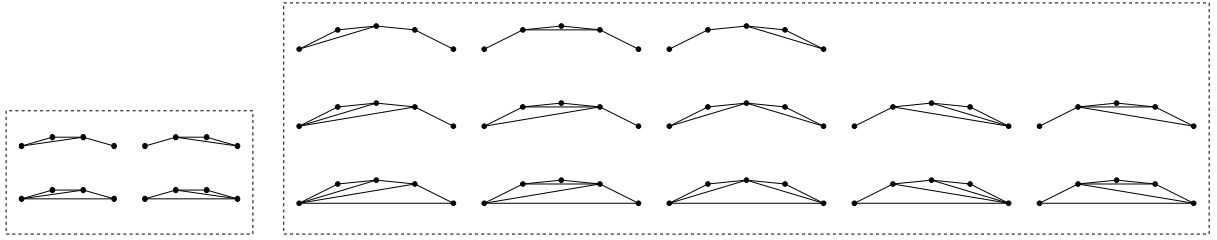}}
\caption{Left: Length reductions for the chains of
$D(n,2^r)$: two of type $1$ and two of type $2$.
Right: Length reductions for $D(n,3^r)$: three of type
$1$, five of type $2$, and five of type $3$.}
\label{f7}
\end{figure}

For $i=0,1,2$, make $\alpha_i r$ reductions of type $i$ in $U$,
for some suitable (constant) values $\alpha_i$ to be determined, with
\begin{equation*} \label{E10}
\sum_{i=1}^3 \alpha_i =1, \ \ \alpha_0,\alpha_1,\alpha_2 \geq 0,
\end{equation*}

According to~\eqref{E5}, choosing one of the three reduction types for
each of the $4$-chains in $U$ can be done in
$$ \Theta^*\left(2^{H_3(\alpha_0,\alpha_1,\alpha_2)r}\right) $$
distinct ways. Once a reduction type $i$ has been chosen for a
specific $4$-chain, it can be implemented in $a_i$ distinct ways.
It follows that the number of distinct resulting chains $U'$ is
\begin{equation*} \label{E11}
\Theta^*\left(2^{H_3(\alpha_0,\alpha_1,\alpha_2)r}\right)
\prod_{i=1}^2 a_i^{\alpha_i r} =
\Theta^*\left( 2^{H_3(\alpha_0,\alpha_1,\alpha_2)r} \cdot
2^{\alpha_1 r} \cdot 2^{\alpha_2 r} \right).
\end{equation*}

So there are that many distinct chains $U'$ which will occur.
Proceed similarly for the lower chain $L$. Observe that once the $\alpha_i$
are fixed, the two resulting sub-chains $U' \subset U$ and $L' \subset L$
have the same length
\begin{equation} \label{E12}
|L'|=|U'|=3r-\alpha_1 r - 2\alpha_2 r =(3-\alpha_1 - 2\alpha_2)r.
\end{equation}
Triangulate the middle part (between $U'$ and $L'$) in any of the
possible ways, according to~\eqref{E7}. Let $N$ denote the total
number of triangulations obtained in this way.
Recall the estimate $t(n,2^r) =\Theta((2^{5/3})^n)$, which yields
$t(3r,2^r) =\Theta((2^{5/3})^{3r})=\Theta(2^{5r})$.
By~\eqref{E7} and~\eqref{E12}, we obtain
\begin{align*} \label{E13}
N &= \Omega^* \left( \left[ (2^{5r})^2 \cdot
2^{2H_3(\alpha_0,\alpha_1,\alpha_2)r} \cdot 2^{2\alpha_1 r} \cdot 2^{2\alpha_2 r}
\cdot 4^{(3-\alpha_1 - 2\alpha_2)r} \right] \right) \nonumber \\
&= \Omega^* \left(
\left[ 2^{5} \cdot 2^{H_3(\alpha_0,\alpha_1,\alpha_2)}
\cdot 2^{\alpha_1} \cdot 2^{\alpha_2} \cdot 2^{(3-\alpha_1- 2\alpha_2)}
\right]^{2r} \right) =
\Omega^* \left( a^n \right),
\end{align*}
where
$$ a=\left[ 2^{8-\alpha_2+H_3(\alpha_0,\alpha_1,\alpha_2)} \right]^{1/3}. $$
The values\footnote{These values have been obtained by numerical
experiments.} $\alpha_0=\alpha_1=0.4$ and $\alpha_2=0.2$ yield
$a= 8.617\ldots$, and $N = \Omega^*(8.617^n)=\Omega(8.61^n)$.

\paragraph{The new $\Omega(8.65^n)$ lower bound: using $D(n,3^r)$.}
We only briefly outline the differences from the previous calculation.
We have $n=|D(n,3^r)|=8r+2$.
For $i=0,1,2,3$, let $a_i$ count the number of distinct ways
reductions can be performed on a convex $5$-chain of $U$.
As illustrated in Fig.~\ref{f7}(right), we have $a_0=1$, $a_0=3$,
and $a_2=a_3=5$. The analogue of~\eqref{E12} is
\begin{equation*} \label{E14}
|L'|=|U'|=4r-\alpha_1 r - 2\alpha_2 r - 3\alpha_3 r =
(4-\alpha_1 - 2\alpha_2 - 3\alpha_3)r.
\end{equation*}
The estimate $t(n,3^r) =\Theta((5^{1/4} \cdot 2)^n)$
now yields $t(4r,3^r) =\Theta((5^{1/4} \cdot 2)^{4r})=\Theta(5^{r} \cdot 2^{4r})$.
The resulting lower bound is
\begin{align*} \label{E15}
N &= \Omega^* \left(
\left[ 5 \cdot 2^{4} \cdot 2^{H_4(\alpha_0,\alpha_1,\alpha_2,\alpha_3)}
\cdot 3^{\alpha_1} \cdot 5^{\alpha_2} \cdot 5^{\alpha_3}
\cdot 2^{(4-\alpha_1- 2\alpha_2 - 3\alpha_3)}
\right]^{2r} \right) \nonumber \\
&= \Omega^* \left(
\left[ 5 \cdot 2^{8-\alpha_1- 2\alpha_2 - 3\alpha_3
+ H_4(\alpha_0,\alpha_1,\alpha_2,\alpha_3)}
\cdot 3^{\alpha_1} \cdot 5^{\alpha_2} \cdot 5^{\alpha_3}
\right]^{n/4} \right) =
\Omega^* \left( a^n \right),
\end{align*}
where
$$ a=\left[5 \cdot 2^{8-\alpha_1- 2\alpha_2 - 3\alpha_3 +
H_4(\alpha_0,\alpha_1,\alpha_2,\alpha_3)}
\cdot 3^{\alpha_1} \cdot 5^{\alpha_2} \cdot 5^{\alpha_3}
\right]^{1/4}. $$
The values $\alpha_0=0.23$, $\alpha_1=0.34$, $\alpha_2=0.29$, $\alpha_3=0.14$,
yield $a= 8.6504\ldots$, and $N = \Omega^*(8.6504^n)=\Omega(8.65^n)$.
The proof of Theorem~\ref{T1} is now complete.
\end{proof}

\paragraph{Remark.}
The next step in this approach, using $D(n,4^r)$, where $n=10r+2$, seems
to bring no further improvement in the bound. Omitting the details, it
gives (with the obvious extended notation):
\begin{align*} 
N &\geq \Omega^* \left(
\left[ 6 \cdot 2^{10-\alpha_1- 2\alpha_2 - 3\alpha_3 - 4\alpha_4
+ H_5(\alpha_0,\alpha_1,\alpha_2,\alpha_3,\alpha_4)}
\cdot 4^{\alpha_1} \cdot 9^{\alpha_2} \cdot 14^{\alpha_3} \cdot 14^{\alpha_4}
\right]^{n/5} \right) =
\Omega^* \left( a^n \right),
\end{align*}
where
$$ a=\left[6 \cdot 2^{10-\alpha_1- 2\alpha_2 - 3\alpha_3 - 4\alpha_4 +
H_5(\alpha_0,\alpha_1,\alpha_2,\alpha_3,\alpha_4)}
\cdot 4^{\alpha_1} \cdot 9^{\alpha_2} \cdot 14^{\alpha_3} \cdot 14^{\alpha_4}
\right]^{1/5}. $$
Numerical experiments suggest that the maximum value of $a$
is smaller than $8.65$, with the corresponding $\alpha$-tuple near
$\alpha_0=0.127$, $\alpha_1=0.254$,
$\alpha_2=0.286$, $\alpha_3=0.222$, $\alpha_4=0.111$, which only yields
$a= 8.6485\ldots$.
Similarly, when taking reflex chains one unit longer
(using $D(n,4^r)$, where $n=12r+2$), numerical experiments suggest that
the maximum value of $a$ is smaller than $8.64$.

\section{Lower bound on the maximum number of
non-crossing spanning trees and forests} \label{sec:trees}

In this section we derive a new lower bound for the number of
non-crossing spanning trees on the double-chain $D(n,0^r)$,
hence also for the maximum number of non-crossing spanning trees
an $n$-element planar point set can have.
The previous best bound, $\Omega(10.42^n)$, is due to Dumitrescu~\cite{D02}.
By refining the analysis of~\cite{D02} we obtain a new bound $\Omega(12.00^n)$.
\begin{theorem}\label{thm:lbst}
For the double chain $D(n,0^r)$, we have
\begin{align*}
\Omega(12.00^n) < \st(D(n,0^r)) & < O(24.68^n), \text{ and} \\
\Omega(12.26^n) < \cf(D(n,0^r)) & < O(24.68^n).
\end{align*}
Consequently, $\st(n)= \Omega(12.00^n)$ and $\cf(n)=\Omega(12.26^n).$
\end{theorem}
\begin{proof}
Similarly to~\cite{D02}, instead of spanning trees, we count
(spanning) forests formed by two trees. One of the trees is
associated with the lower chain $L$ and is called \emph{lower tree},
while the other tree is associated with the upper chain $U$ and is
called \emph{upper tree}. Since the two trees can be connected in at
most $O(n^2)$ ways, it suffices to bound from below the number of forests of two
such trees.
\begin{figure}[htbp]
\centerline{\includegraphics[width=.6\columnwidth]{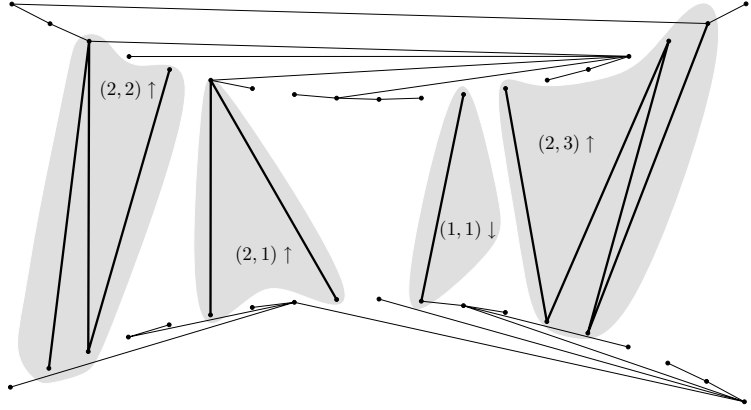}}
\caption{A double chain with lower and upper tree and four bridges.}
\label{fig:pairings}
\end{figure}

Fig.~\ref{fig:pairings} shows an example. We count only special kinds of forests:
no edge of the lower tree connects two vertices of the upper chain, and similarly,
no edge of the upper tree connects two vertices of the lower chain.
We call the connected components of the edges between $U$ and $L$
\emph{bridges}. For the class of forests we consider, bridges are
subtrees of the lower or the upper tree. A bridge is called an \emph{$(i,j)$-bridge}
if it has $i$ vertices in $L$ and $j$ vertices in $U$.
Every bridge is part of either the upper or the lower tree.
In the first case we say that the bridge is oriented {\em upwards} and in the latter
case that it is oriented {\em downwards}. Since edges cannot cross, the bridges have
a natural left-to-right order. Fig.~\ref{fig:pairings} shows four
bridges, where the first bridge is an upward oriented
$(2,2)$-bridge. We consider only bridges
$(i,j)$, with $1\leq i,j \leq z$, for some fixed positive integer $z$.
For $z=1$, our analysis coincides with the one in~\cite{D02}, and we first
re-derive the lower bound of $\Omega(10.42^n)$ given there.
Successive improvements are achieved by considering lager values of $z$.

Let $m=n/2$ be the number of points in each chain. The distribution of
bridges is specified by a set of parameters $\alpha_{ij}$, to be determined later,
where the number of $(i,j)$-bridges is $\alpha_{ij} m$. To simplify
further expressions we introduce the following wildcard notation:
\[ \alpha_{i*}=\sum_{k=1}^z \alpha_{ik}, \quad \alpha_{*j}=\sum_{k=1}^z
\alpha_{kj},\quad \text{and} \quad
\alpha_{**}=\sum_{k=1}^z \alpha_{*k}=\sum_{k=1}^z \alpha_{k*}.\]
A vertex is called a \emph{bridge vertex}, if it is part of
some bridge, and it is a \emph{tree vertex} otherwise.
We denote by $\alpha_L m$ the number of bridge vertices along $L$,
and by $\alpha_U m$
the number of bridge vertices along $U$, we have
\[ \alpha_L= \sum_{k=1}^z k\alpha_{k*}, \quad \text{and} \quad \alpha_U =
\sum_{k=1}^z k\alpha_{*k}.\]

To count the forests we proceed as follows. We first count the distributions of the
vertices that belong to bridges on the lower ($N_L$) and upper chain ($N_U$).
We then count the number of different ways in which the bridges can be realized ($N_\text{bridges}$)
and the number of ways in which the bridges can be connected to the two trees ($N_{\text{links}}$).
Finally, we estimate the number of the trees within the two chains
($N_{\text{trees}}$). All these numbers are parameterized by the
variables $\alpha_{ij}$.

Consider the feasible locations of bridge vertices at the
lower chain. We have ${m \choose \alpha_L m}$ ways to choose the bridge vertices
that are in $L$. Every bridge vertex belongs to some $(i,j)$-bridge. The vertices
of the bridges cannot interleave, thus we can describe the
configuration of bridges by a sequence of $(i,j)$ tuples that denotes
the appearance of the $\alpha_{**}m$ bridges from left to right on $L$.
There are ${ \alpha_{**} m \choose \alpha_{11}m , \alpha_{12}m ,
\dotsc , \alpha_{zz}m}$ such sequences. This give us a total of
\[N_L:={m \choose \alpha_L m} { \alpha_{**} m \choose \alpha_{11} m,
\alpha_{12}m, \dotsc ,
\alpha_{zz}m}=\Theta^*\left(2^{H(\alpha_L)m+
\alpha_{**}H_{\left(z^2\right)}(\alpha_{11}/\alpha_{**},\dotsc,\alpha_{zz}/\alpha_{**})m}\right)\]
such ``configurations'' of bridge vertices along $L$.

We now determine how many options we have to place the bridge
vertices on $U$. Since we have already specified the sequence of
the $(i,j)$-bridges at the lower chain, all we can do is to select the bridge
vertices in $U$. This gives
\[N_U:={m \choose \alpha_U m}=\Theta^*\left(2^{H(\alpha_U) m}\right)\]
possibilities for the configuration on $U$.

We now study in how many ways the bridges can be added to the two
trees. Since all bridges are subtrees, we can link one of the bridge
vertices with the lower or upper tree. From this perspective the whole
bridge acts like a super-node in one of the trees. The orientation
of each bridge determines the tree it belongs to: upwards
bridges to the upper tree, and downwards bridges to the lower tree.
For every pair $(i,j)$ we orient half of the $(i,j)$-bridges upwards
and half of them downwards. To glue the bridges to the trees we have
to specify a vertex that will be linked to one of the trees.
Depending on the orientation of the $(i,j)$-bridge, we have $i$
candidates for a downwards oriented bridge and $j$ candidates
for an upward oriented bridge. In total we have
\[N_{\text{links}}:=
\prod_{i,j} \binom{\alpha_{i,j}m}{\alpha_{i,j}m/2} \left( i^{\alpha_{ij} /2} j^{\alpha_{ij} /2} \right)^m
=\prod_{i,j} \Theta^*\left(2^{\alpha_{i,j}m}\right) \left( ij\right)^{\alpha_{ij}m/2}
=\Theta^*\left(2^{\alpha_{**}m}\right)\prod_{i,j} \left( ij\right)^{\alpha_{ij}m/2}
 \]
ways to link the bridges with the trees.

Until now we have specified which vertices belong to which type of
bridges, the orientation of the bridges, and the vertex where the
bridge will be linked to its tree. It remains to count the number of
ways to actually ``draw'' the bridges. Let us consider an
$(i,j)$-bridge. All of its edges have to go from $L$ to $U$ and the bridge
has to be a tree. The number of such trees equals the number of
triangulations of a polygon with point set $\{(k,0) \; | \; 0\leq k
\leq i\} \cup \{ (k,1)\; | \; 0\leq k \leq j\}$. By deleting the
edges along the horizontal lines $y=0$ and $y=1$, we define a bijection
between these triangulations and the combinatorial types of $(i,j)$-bridges.
The number of triangulations is now easy to express similarly to Equation~\eqref{E7}:
We have $i+j-2$ triangles, and each triangle is adjacent to a horizontal edge
along either $y=0$ or $y=1$, where exactly $i-1$ triangles are adjacent to line $y=0$.
In total we have $B_{ij}:={ i+j-2 \choose i-1 }$ different
triangulations and therefore
we can express the number of different bridges by
\[N_{\text{bridges}}=\Big( \prod_{ij} B_{ij}^{\alpha_{ij}} \Big)^m.\]

\begin{figure}[htbp]
\centerline{ \includegraphics[width=.6\columnwidth]{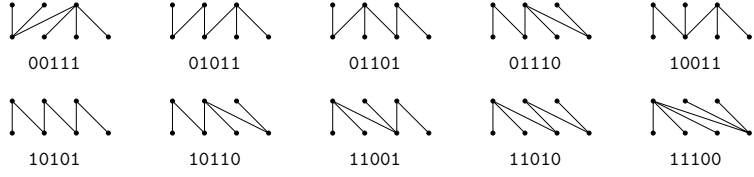}}
\caption{All $B_{34}=10$ combinatorial types of
  $(3,4)$-bridges. If an edge differs from its predecessor at the top
  we write a
{\tt 0}, otherwise a {\tt 1}. We obtain a bijection between the
bridges and sequences with three {\tt 1}s and two {\tt 0}s.}
\label{fig:bridges}
\end{figure}

Observe that the upper and the lower trees are trees on a convex point set.
By considering the bridges as super-nodes, we treat the lower chain as
a convex chain of $n_L$ vertices. Similarly, we think of the upper chain as
a convex chain with $n_U$ vertices.
We have
\begin{align*}
n_U=\Big(1-\sum_{k=1}^n
\frac{2k-1}{2} \alpha_{k*}\Big)m, \qquad \text{and} \qquad
n_L=\Big(1-\sum_{k=1}^n \frac{2k-1}{2} \alpha_{*k}\Big)m.
\end{align*}
 (Notice that the bridges take away all of its vertices, except one,
depending on the orientation.) Since the number of non-crossing
spanning trees on an $n$-element convex point set equals
$\Theta^*((27/4)^n)$~\cite{FN99}, the number of
spanning trees within the two chains is given by
\[N_{\text{trees}}=O^*\Big( \Big( 27/4 \Big) ^{n_L+n_U}\Big).\]
To finish our analysis we have to find the optimal parameters
$\alpha_{ij}$ such that
\begin{equation}\label{eq:LBST} \st(D(n,0^r))=\Omega^*\left( N_L
\cdot N_U \cdot N_{\text{bridges}} \cdot N_{\text{links}} \cdot
N_{\text{trees}} \right)
\end{equation}
is maximized.

Assume that we picked all $\alpha_{ij}$ values, except one
$\alpha_{uv}$, and let $g_{uv}(\cdot)$ be the bound of~\eqref{eq:LBST}
in terms of this only unspecified parameter. It can be observed that
$g_{uv}\equiv g_{vu}$. Hence we can replace all occurrences of
$\alpha_{uv}$ in~\eqref{eq:LBST} with $u>v$ by $\alpha_{vu}$.
To solve the optimization problem we restrict ourselves to a few small
values of $z$. Let us start with the
case $z=1$. Here our analysis coincides with the one in~\cite{D02}. We
have only one parameter $\alpha_{11}$, which maximizes \eqref{eq:LBST} at
$4/(4+3 \sqrt{6})=0.35\ldots$ and yields a lower bound of $\Omega(10.42^n)$. A
first improvement is achieved by considering $z=2$. In this case
\eqref{eq:LBST} is bounded by
$\Omega(11.61^n)$, which is attained at $\alpha_{11}= 0{.}18,
 \alpha_{12}=0{.}055$, $\alpha_{22}=0{.}032$. For the next
step $z=3$, we obtain a lower bound of $\Omega(11.89^n)$, attained at
$\alpha_{11}=0{.}15$, $\alpha_{12}=0{.}043$,
$\alpha_{13}= 0{.}010$, $\alpha_{22}=
0{.}023$, $\alpha_{23}=0{.}0085$, and $\alpha_{33}=
0{.}0040$. The optimal solution for  $z=4$ are listed in
Table~\ref{tab:z4}. The induced bound equals  $\Omega(12.00^n)$.
By further increasing $z$ we get improved bounds. The best bound was
computed for $z=8$, namely $\Omega(12.0026^n)$. The parameters
realizing this bound can be found in Table~\ref{tab:z8}. For larger
$z$ we were not able to complete the numeric computations and got
stuck at a local maximum, whose induced bound was smaller than
$\Omega(12.0026^n)$.
All optimization problems were solved with
computer algebra software.
\begin{table}[htdp]
\begin{minipage}[b]{0.5\linewidth}\centering
\centering
\begin{tabular}{c|| l l l l}
$\alpha_{ij}$ & $j=1$& $j=2$ & $j=3$ & $j=4$ \\
\hline \hline
$i\!=\!1$ & $0.149$ & $0.0403$ & $0.00945$ & $0.00208$ \\
$i\!=\!2$ & - & $0.0218$ & $0.00767$ & $0.00226$ \\
$i\!=\!3$ & - & - & $0.00359$ & $0.00132$ \\
$i\!=\!4$ & - & - & - & $ 0.00058$
\end{tabular}
\end{minipage}
\begin{minipage}[b]{0.5\linewidth}
\centering
\begin{tabular}{c|| l l l l}
$\alpha_{ij}$ & $j=1$& $j=2$ & $j=3$ & $j=4$ \\
\hline \hline
$i\!=\!1$ & $0.151$ & $0.0382$ & $0.00835$ & $0.00172$ \\
$i\!=\!2$ & - & $0.0192$ & $0.00632$ & $0.00226$ \\
$i\!=\!3$ & - & - & $0.00276$ & $0.00094$ \\
$i\!=\!4$ & - & - & - & $0.00039$
\end{tabular}
\end{minipage}
\caption{On the left: Parameters that attain the maximum
  $\Omega(12.00^n)$ for
$\st(D(n,0^r))$ for $z=4$. On the right: Parameters that attain the
  maximum $\Omega(12.26^n)$ for
$\cf(D(n,0^r))$ for $z=4$.}
\label{tab:z4}
\end{table}

\begin{table}
\centering
\begin{tabular}{l||llllllll}
 $\alpha_{ij}$ & $i=1$ & $i=2$ & $i=3$ & $i=4$ & $i=5$ & $i=6$ & $i=7$ & $i=8$ \\
  \hline
\hline
 $j=1$ & 0.144 & 0.0389 & 0.00908 & 0.001994 & 0.000422 & 0.0000856 & 0.0000152 & $1.76\cdot 10^{-6}$ \\
 $j=2$ & - & 0.0209 & 0.00733 & 0.00214 & 0.000569 & 0.000140 & 0.0000313 & $5.12\cdot 10^{-6} $\\
 $j=3$ & - & - & 0.00342& 0.00125 & 0.000397 & 0.000113 & 0.0000290 & $5.50\cdot 10^{-6}$ \\
 $j=4$ & - & - & - & 0.000548& 0.000202 & 0.0000655 & 0.0000181 & $3.34\cdot 10^{-6}$ \\
 $j=5$ & - & - & - & - & 0.0000845 & 0.0000298 & $8.33\cdot 10^{-6}$ & $1.25\cdot 10^{-6}$ \\
 $j=6$ & - & - & - & - & - & 0.0000107 & $2.44\cdot 10^{-6}$ & $5.10\cdot 10^{-7}$ \\
 $j=7$& - & - & - & - & - & - & $1.09\cdot 10^{-6}$ & $1.31\cdot 10^{-7}$ \\
 $j=8$ & - & - & - & - & - & - & - & $6.97\cdot 10^{-8}$
\end{tabular}
\caption{Parameters that realize the maximum
  $\Omega(12.0026^n)$ for
$\st(D(n,0^r))$ for $z=8$.}
\label{tab:z8}

\end{table}

\paragraph{Cycle-free graphs.}
The same approach can be used to bound the number of cycle-free
graphs (\ie {\em forests}) on the double chain. To this end, we update
Equation~\eqref{eq:LBST} by substituting the quantity $N_{\text{trees}}$ with
\[N_{\text{forests}}=\Omega^*\left(8.22469^{n_L+n_U}\right),\]
which is obtained from the bound for the number of forests on a convex
point set~\cite{FN99}. Since we picked the bridges such that they introduce no cycles,
all graphs considered by our scheme are cycle-free graphs. When the
super-node (that represents a bridge) is a singleton in the forest of
the chain, the combined forest could be constructed by different
gluings of that bridge. For this reason we have to ignore the
possibilities that were covered by $N_{\text{links}}$ to avoid
over-counting. Thus we end up with
\[
 \cf(D(n,0^r))=\Omega^*\left( N_L
\cdot N_U \cdot N_{\text{bridges}} \cdot
N_{\text{forests}} \right).
\]

For $z=2$, our method yields the bound $\Omega(11.94^n)$
attained at $\alpha_{11}= 0{.}18$, $\alpha_{12}=0{.}039$,
and $\alpha_{22}=0{.}021$; and for $z=3$ the bound
$\Omega(12.16^n)$ attained at $\alpha_{11}=0{.}12$,
$\alpha_{12}=0{.}031$, $\alpha_{13}= 0{.}0080$, $\alpha_{22}=
0{.}016$, $\alpha_{23}=0{.}0061$, and $\alpha_{33}= 0{.}0031$.
The bound for  $z=4$, namely $\Omega(12.26^n)$, is attained at the
parameters  listed in Table~\ref{tab:z4}.
The best bound was obtained for $z=9$. In this case the parameters in
Table~\ref{tab:z9} imply a bound of $\Omega(12.2618^n)$. Numeric
computations for larger numbers of $z$ gave no improvement
on the bound for $\cf(D(n,0^r))$.
The previous best lower bound by Aichholzer~\etal was
$\Omega(11.62^n)$.
\begin{table}
\centering
\begin{tabular}{l||lllllllll}
 $\alpha_{ij}$ & $i=1$ & $i=2$ & $i=3$ & $i=4$ & $i=5$ & $i=6$ & $i=7$ & $i=8$& $i=9$ \\
  \hline
\hline
$j=1$ 	 & 0.11	 &0.028	&0.0069&	0.0017&	0.00042&	0.00010&	0.000024&	$5.3\cdot10^{ -6}$&	$1.2\cdot10^{ -6}$\\
$j=2$	& -	&0.014	&0.0051&	0.0017&	0.00052&	0.00015&	0.000042	&0.000011&	$2.7\cdot10^{ -6}$\\
$j=3$	& -	& -&	0.0025	&0.0010&	0.00038&	0.00013&	0.000042&	0.000012&	$3.4\cdot10^{ -6}$\\
$j=4$	& -	& -	& -	&0.00051	&0.00022	&0.000086	&0.000030	&0.000010	&$3.1\cdot10^{ -6}$\\
$j=5$	& -	& -	& -	& -&	0.00011	&0.000047	&0.000018	&$6.5\cdot10^{ -6}$	&$2.2\cdot10^{ -6}$\\
$j=6$	& -	& -	& -	& -	& -&	0.000024	&$9.4\cdot10^{ -6}$&	$3.7\cdot10^{ -6}$&	$1.4\cdot10^{ -6}$\\
$j=7$	& -	& -	& -	& -	& -	& -	&$5.6\cdot10^{ -6}$&	$1.8\cdot10^{ -6}$	&$7.5\cdot10^{ -7}$\\
$j=8$	& -	& -	& -	& -	& -	& -	& -	&$1.5\cdot10^{ -6}$	&$3.9\cdot10^{ -7}$\\
$j=9$	& -	& -	& -	& -	& -	& -	& -	& -	&$4.4\cdot10^{ -7}$\\
\end{tabular}
\caption{Parameters that realize the maximum
  $\Omega(12.0026^n)$ for
$\cf(D(n,0^r))$ for $z=9$.}
\label{tab:z9}
\end{table}

Notice that our analysis relies on numeric computations. We were able
to compute the optimal parameters up to $z=8$ (for spanning trees) and up
to $z=9$ (for cycle-free graphs). Table~\ref{tab:zz} lists the computed
bounds with respect to the chosen parameter $z$. No significant
improvements can be expected by solving the induced maximizations
problems for larger values of $z$.
\begin{table}
\centering
\begin{tabular}{l||lllllllll}
  & $z=1$ & $z=2$ & $z=3$ & $z=4$ & $z=5$ & $z=6$ & $z=7$ & $z=8$ & $z=9$ \\
  \hline
\hline
 $\st$ & 10.424 & 11.611 & 11.899 & 12.004 & 11.952 & 11.998 & 12.002 & 12.002 & - \\
 $\cf$ &  11.092 & 11.944 & 12.169 & 12.260 & 12.251 & 12.258 & 12.260 & 12.261 & 12.261 \\
 \end{tabular}
\caption{Bases of the asymptotic exponential bounds for the
  number of spanning trees $\st$ and the number of cycle-free graphs
  $\cf$ for the point set  $D(n,0^r))$ in terms of the parameter $z$.}
\label{tab:zz}
\end{table}

For comparison, we compute an easy upper bound for both
$\st(D(n,0^r))$ and $\cf(D(n,0^r))$.
Every forest $F$ on the double chain splits into three parts: the
induced subgraph on $L$ ($F_L$),
the induced subgraph on $U$ ($F_U$) and the remaining part given by
the bridge edges ($F_B$). The graphs $F_L$, $F_U$, and $F_B$ are
non-crossing forests.
Since the number of forests on $n$ points in convex position is
$O(8.225^n)$, there are $O(8.225^n)$ different graphs $F_L \cup F_U$
(we rounded up to avoid the $O^*$ notation).
%
It remains to bound the number of graphs $F_B$. Every graph counted in
$F_B$ can be turned into a triangulation by deleting the ``non-bridge
vertices'' and adding the appropriate edges on $L$ and $U$. Let $k$ be
the number of bridge vertices on each chain. We have ${m \choose k}^2$
ways to select the bridge vertices. By Equation~\eqref{E7} we have
$\Theta^*(4^{k})$ triangulations of the bridge vertices. In total we
can bound the number of graphs $F_B$ from above by
 $\sum_{k=0}^m { m \choose k }^2 4^k.$
To bound the exponential growth of this sum we find the dominating
summand. Notice that ${ m \choose k' }^2 4^{k'}$ is maximized when
${ m \choose k' } 2^{k'}$ is maximized.
Since $\sum_{k=0}^m { m \choose k } 2^k=3^m$, we have $ { m \choose k' } 2^{k'}<3^m$,
and so
\[{ m \choose k }^2 4^{k} \leq { m \choose k' }^2 4^{k'}<9^m=3^n.\]

Therefore the number of cycle-free non-crossing graphs on $D(n,0^r)$ is
bounded from above by
$O( n \cdot 3^n \cdot 8.225^n)=O(24.68^n)$,
which is also an upper bound for the number of non-crossing spanning
trees on $D(n,0^r)$.
\end{proof}

\section{Upper bound for the number of non-crossing spanning cycles}
\label{sec:cycles}

Newborn and Moser~\cite{NM80} asked what is the maximum number
of non-crossing spanning cycles for $n$ points in the plane, and they proved
that $\Omega((10^{1/3})^n)\leq \sc(n) \leq O(6^n \lfloor\frac{n}{2}\rfloor !)$.
The first exponential upper bound $\sc(n) \leq 10^{13n}$, obtained by
Ajtai~\etal\cite{ACNS82}, has been followed by a series of
improved bounds, \eg see~\cite{BKK+07,D99,SW06}; a more comprehensive history
can be found in~\cite{Dem10}. The current best lower bound
$\sc(n) \geq 4.462^n$ is due to Garc\'{\i}a~\etal\cite{GNT00}.
The previous best upper bound $O(70.22^n)$ follows from combining the
upper bound $30^{n/4}$ of Buchin~\etal~\cite{BKK+07} on the number of
spanning cycles in a triangulation with a new upper bound of $\tr(n) \leq 30^n$
on the number of triangulations by Sharir and Sheffer~\cite{SS10}
($30^{5/4}=70.21\ldots$).
The bound by Buchin~\etal\cite{BKK+07} cannot be improved much further,
since they also present triangulations with $\Omega(2.0845^n)$ spanning cycles.
However, the approach of multiplying $\tr(n)$ with the maximum number
of spanning cycles in a triangulation seems rather weak, since it 
potentially counts some spanning cycles many times.

Let $C$ be a spanning cycle. We say that  $C$ has a \emph{support} of
$x$, and write $\supp(G)=x$ if $C$ is contained in $x$ triangulations
of the point set.

Observe that a spanning cycle $C$ will be counted $\supp(C)$ times in
the preceding bound.
To overcome this inefficiency, we use the concept of \emph{pseudo-simultaneously
flippable edges} (ps-flippable edges, for short), introduced in~\cite{HSSTW11}.
A set $F$ of edges in a triangulation is \emph{ps-flippable} if after
deleting all edges in $F$, the bounded faces are convex and jointly
tile the convex hull of the points.
One can obtain a lower bound for the support of a spanning cycle $C$
in terms of the number of ps-flippable edges that are {\em not} in $C$.
The following two properties of ps-flippable edges derived
in~\cite{HSSTW11} are relevant to us:
\begin{enumerate}[(i)]
\item Every triangulation on $n$ points contains at least
  $\frac{n}{2}-2$ ps-flippable edges. \label{pro:lower}
\item Let $C$ be a spanning cycle. Consider a triangulation $T$ that contains $C$, and
       has a set $F$ of ps-flippable edges. If $j$ or more edges of
       $F$ are not contained in $C$, then $\supp(C) \ge 2^j$.
\end{enumerate}

We now use ps-flippable edges to improve the upper bound for the
number of spanning cycles.
\begin{theorem}\label{thm:polygonize}\label{th:triVSsc}
$\displaystyle \sc(n) = O \left(68.62^n \right). $
\end{theorem}
\begin{proof}
Let $P$ be a set of $n$ points in the plane. Assume first that $n$ is even.
The exact value of $\sc(P)$ is
\begin{equation} \label{eq:optimal}
\sc(P) = \sum_{T} \sum_{C \subset T} \frac{1}{\supp(C)},
\end{equation}
where the first sum is taken over all triangulations of $P$, and the second sum is
taken over all spanning cycles contained in $T$.

Consider a triangulation $T$ over the point set $P$. By property~(\ref{pro:lower}),
there is a set $F$ of $n/2-2$ ps-flippable edges in $T$. We wish to bound the number
of spanning cycles that are contained in $T$ and use exactly $k$ edges from $F$.
The support of any such cycle is at least $2^{|F|-k}$ by property~(ii).
There are $\binom{|F|}{k}<\binom{n/2}{k}$ ways to choose $k$ edges from $F$.

Since $n$ is even, every spanning cycle in $T$ is the union of two
disjoint perfect matchings in $T$. Instead of spanning cycles containing
exactly $k$ edges of $F$, we count the number of pairs of disjoint
perfect matchings, $M_1$ and $M_2$, that jointly contain exactly $k$ edges of
$F$. Assume, without loss of generality, that $M_2$ contains at least as
many edges from $F$ as $M_1$.
That is, $M_1$ contains at most $k/2$ edges from $F$, and $M_2$ at least $k/2$.

As observed by Buchin~\etal~\cite{BKK+07}, a plane graph with $v$
vertices and $e$ edges contains at most $O(( 2e/v )^{v/4})$ perfect
matchings. By Euler's formula,
a triangulation on $n$ points has at most $3n - 6$ edges.
Moreover, we know that both, $M_1$ and $M_2$, have to avoid the $|F|-k$
edges from $F$ that were not selected, thus it is safe to remove
$|F|-k$ edges from the triangulation. Since
$$ (3n-6)-(n/2-2-k)=5n/2+k-4 \leq 5n/2+k, $$
the remaining graph has at most $5n/2+k$
edges and by the previous expression of Buchin~\etal, there are
$O\left((5+2k/n)^{(n/4)}\right)$ ways to choose $M_1$.

After choosing $M_1$, we know of at least $k/2$ edges of $F$ which
participate in $M_2$ (the edges that were chosen from $F$ and do not participate in $M_1$).
We can remove the endpoints of these edges (together with all
incident edges). The resulting graph has at most $n-k$ vertices. We can also remove
all remaining edges of $F$ and $M_1$, since they cannot be in $M_2$. It can be easily
checked that overall we removed $n-2-k/2$ edges (recall that $|M_1| = n/2$,
and that it can hold at most $k/2$ edges of $F$). Thus, given a
matching $M_1$, the number of candidates for $M_2$ is
$$ O\left(\left( \frac{2(3n-(n-k/2))}{n-k} \right)^{(n-k)/4} \right)
= O\left(\left( \frac{4n+k}{n-k} \right)^{(n-k)/4} \right) .$$

We conclude that the number of spanning cycles containing exactly $k$
edges of $F$ is
\begin{equation*} \label{eq:countK}
O \left(\binom{n/2}{k} \cdot (5+2k/n)^{(n/4)} \cdot \left(
\frac{4n+k}{n-k} \right)^{(n-k)/4} \right).
\end{equation*}

When $k$ is small, \ie $k \leq an$, for some $a \in (0,1/2)$, it is
better to use the bound $30^{n/4}$ instead
(which bounds the total number of spanning cycles in $T$).
Substituting these bounds into~\eqref{eq:optimal} yields
\[ \sc(P) = \sum_{T} \left( \sum_{k = 0}^{an} O \left(
\frac{30^{(n/4)}}{2^{n/2-k}}\right) + \ns  \sum_{k = an}^{n/2}O \left(
\binom{n/2}{k} \cdot \left( \frac{4n+k}{n-k} \right)^{(n-k)/4} \ns
\cdot  \frac{(5+2k/n)^{(n/4)}}{2^{n/2-k}} \right) \right). \]
The maximum is attained for $a \approx 0.466908$, which implies
\[ \sc(P) = \sum_{T} O\left(2.28728^n \right) = \tr(P) \cdot O\left(2.28728^n \right). \]
The upper bound $\sc(n) = O(68.62^n)$ is immediate by substituting
the upper bound $\tr(P)<30^n$ from~\cite{SS10} into the above expression.

Finally, the case where $n$ is odd can be handled as follows. Create a new point $p$
outside of $\conv(P)$, and put $P' = P \cup \{p\}$.
It can be shown that $\sc(P) \le \sc(P')$ by mapping every spanning cycle of $P$ to
a distinct spanning cycle of $P'$. Given a non-crossing cycle $C$ of $P$,
 $p$ can be connected to the two endpoints of some
edge of $C$ without crossings~\cite[Lemma~2.1]{HKRT08}. We can then apply the
above analysis for $P'$, obtaining the same asymptotic bound.
\end{proof}

\section {Weighted geometric graphs} \label{sec:weighted}
In this section we derive bounds on the maximum multiplicity of various weighted
geometric graphs (weighted by Euclidean length).

\subsection {Longest perfect matchings}
Let $n$ be even, and consider perfect matchings on a set of $n$
points in the plane. It is easy to construct $n$-element point sets
(no three of which are collinear) with an exponential number
of {\em longest} perfect matchings:~\cite{D02} gives constructions with
$\Omega(2^{n/4})$ such matchings.
Moreover, the same lower bound can be achieved with yet another restriction,
convex position, imposed on the point set; see~\cite{D02}.
Here, we present constructions with an exponential number of {\em maximum}
(longest) non-crossing perfect matchings.

\begin{theorem} \label{T2}
For every even $n$, there exist $n$-element point sets with
at least $2^{\lfloor n/4\rfloor}$ longest non-crossing
perfect matchings. Consequently, $\pm_{\rm max}(n) =\Omega(2^{n/4})$.
\end{theorem}
\begin{proof}
Assume first that $n$ is a multiple of $4$.
Let $S_4=\{a,b,c,d\}$ be a $4$-element point set such that segment
$ab$ is vertical, $cd$ lies on the orthogonal bisector of $ab$ (hence,
$|ac|=|bc|$ and $|ad|=|bd|$), $|ab|=|cd|=\frac{1}{n}$ and
$\min\{|ac|,|ad|\}=|ac|=|bc|=2n$.
Then $S_4$ has two maximum matchings, $\{ac,bd\}$ and $\{ad,bc\}$, each of which
has length at least $4n$.
Let the $n$-element point set $P$ be the union of $n/4$ translated copies of $S_4$
lying in disjoint horizontal strips such that the copies of $a$ are
almost collinear, all the copies of points $a$ and $b$ lie in a disk
of unit diameter, and all the copies of points $c$ and $d$ lie in a disk
of unit diameter; see Fig.~\ref{fig:longmatching}.
\begin{figure}[htbp]
\centerline{\epsfxsize=5.5in \epsffile{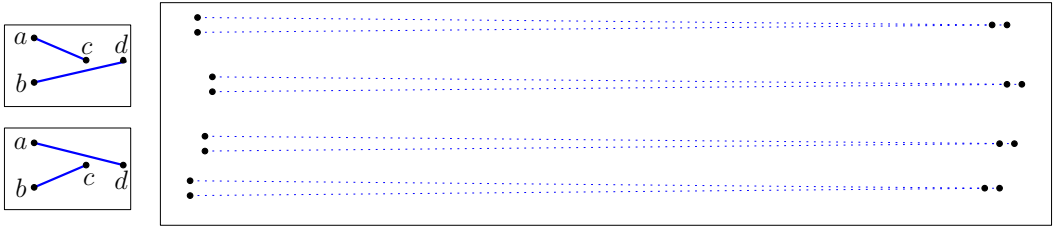}}
\caption{Left: Two possible maximum matchings for the point set
  $S_4=\{a,b,c,d\}$. Right: A set of $n=16$ points that admit $2^4$
  maximum non-crossing perfect matchings.}
\label{fig:longmatching}
\end{figure}

If we combine the maximum matchings of all copies of $S_4$, then we obtain
$2^{n/4}$ non-crossing perfect matchings of $P$. All these matchings
have the same length,
which is at least $\frac{n}{4}\cdot 4n=n^2$. We show that this is the
{\em maximum} possible length of a non-crossing perfect matching of $P$.
Let $L$ be the set of the translated copies of points $a$ and $b$; and
let $R$ be the set of the translated copies of points $c$ and $d$.
Then the diameter of $L$ (resp., $R$) is at most $1$, by
construction. The length of any edge between $L$ and $R$ is at most
$2n+1$, while all other edges have length at most $1$.
If a matching has $k$ edges between $L$ and $R$, then its length is at most
$k(2n+1)+(\frac{n}{2}-k)=2kn+\frac{n}{2}$, which is less than $n^2$ for
$k<n/2$. Hence a maximum matching on $P$ is a bipartite graph between $L$ and
$R$. To avoid crossings, every edge in such a bipartite graph
must connect points in the same copy of $S_4$.

If $n$ is a multiple of $4$ plus $2$, the above construction is
modified by adding a pair of points at distance $2n$, placed
horizontally above all copies of $S_4$. This pair of points must be
matched in any maximum non-crossing perfect matching of $P$.
In both cases, $P$ admits at least $2^{\lfloor n/4\rfloor}$ longest non-crossing
perfect matchings, as required.
\end{proof}

\subsection {Non-crossing spanning trees}
For minimum spanning trees,~\cite{D02} gives constructions
for $n$-element point sets that admit
$\Omega(2^{n/4})$ such trees, and moreover, this lower bound
can be attained with points in convex position. All these
constructions give non-crossing spanning trees.
For maximum spanning trees,~\cite{D02} gives the following
construction: start with two points,
$a$ and $b$, and suitably place the remaining $n-2$ points on the
perpendicular bisector of segment $ab$. While this configuration
admits $\Omega(2^{n})$ maximum non-crossing spanning trees
(which have maximum weight over all spanning trees), it uses (a large number of)
collinear points. Next we show that an exponential bound, $\Omega(2^n)$, can be
achieved without allowing collinear points, and moreover, with points
in convex position.

\begin{theorem} \label{T3}
The vertex set of a regular convex $n$-gon admits
$\Omega(2^n)$ longest non-crossing spanning trees.
Consequently, $\st_{\rm max}(n)=\Omega(2^n)$.
\end{theorem}

Before proving the theorem we introduce some notation, which we will
also use in Section~\ref{sec:tours}.
For a point set in convex position the \emph{span} of an edge $xy$
is the smallest number of convex hull edges one has to traverse
when going from $x$ to $y$ along the convex hull (in clockwise or counterclockwise direction).
The weight of a tree $T$, denoted as $L(T)$, is the sum of its edge lengths.

\paragraph{Proof of Theorem~\ref{T3}.}
Let $P$ be the vertex set of a regular $n$-gon inscribed in a circle of unit radius.
Let $p\in P$ be an arbitrary element of $P$, and let $S_p$ be the star
centered at $p$ (consisting of segments connecting $p$ to the other
$n-1$ points). We first prove a counterpart for spanning trees (Lemma~\ref{L1})
of Alon~\etal\cite[Lemma 2.2]{ARS95}, established for non-crossing matchings.
Our argument is similar to that used in~\cite{ARS95}.

\begin{lemma} \label{L1}
A longest non-crossing spanning tree on $P$ has weight $L(S_p)$.
\end{lemma}
\begin{proof}
Write $n=2k$ if $n$ is even, or $n=2k+1$ if $n$ is odd.
Consider an arbitrary non-crossing spanning tree $T$ of $P$.
Notice that the span of any edge is an integer between $1$ and $k$.
For $i=1,\ldots , k$, let $N_i$ denote the number of edges of $T$ whose
span is at least $i$. We need the following claim.

\smallskip
\noindent{\em Claim.}
For every $i=1,\ldots , k$, we have $N_i \leq n-2i+1$.

\smallskip
\noindent{\em Proof of Claim:}
The claim trivially holds for $N_i \leq 1$, so we can assume that
$N_i \geq 2$. One can easily check (since $T$ is non-crossing) that
there exist two edges $p_1 p_2$ and $p_3 p_4$ in $T$ (where $p_2$
could coincide with $p_3$), each of span at least $i$, with the
following properties:
(a) $p_1,p_2,p_3,p_4$ appear in this (clockwise) order on the circle,
and (b) no other edge of $T$ with span at least $i$ has an endpoint
among the points between $p_1$ and $p_2$ or between $p_3$ and $p_4$.
There are at least $i-1$ points on each of the two open circular arcs
defined by $p_1 p_2$ and $p_3 p_4$. Hence all edges
of $T$ with span at least $i$ are induced by at most $n-2(i-1)$ points.
Since this set of edges forms a forest, it has no more than
$n-2(i-1)-1=n-2i+1$ elements, as required.
\qed

\medskip
To finalize the proof of the lemma, we show that $L(T) \leq L(S_p)$.
For $i=1,\ldots ,k$, let $\ell_i$ denote the (Euclidean) length of
an edge with span $i$. Note that $\ell_1<\ell_2<\ldots<\ell_k$, and
define $\ell_0=0$. Put also $N_{k+1}=0$, and observe that the number of edges
of span $i$ is equal to $N_i - N_{i+1}$. Consider first the case $n=2k+1$.
Straightforward algebraic manipulation gives
\begin{align*}
L(T) &= \sum_{i=1}^k (N_i - N_{i+1}) \ell_i =
\sum_{i=1}^k N_i (\ell_i -\ell_{i-1}) \\
&\leq \sum_{i=1}^k (n-2i+1) (\ell_i -\ell_{i-1}) =
2 \sum_{i=1}^k \ell_i = L(S_p).
\end{align*}
Similarly, for $n=2k$ we get
$$ L(T) \leq 2 \sum_{i=1}^{k-1} \ell_i + \ell_k = L(S_p). $$
This concludes the proof of the lemma.
\end{proof}

To finalize the proof of Theorem~\ref{T3}, it remains to show that $P$ admits
$\Omega(2^{n})$ non-crossing spanning trees of weight $L(S_p)$.
Denote the points in $P$ by $p_1, p_2, \ldots, p_n$ in clockwise order.
We construct a family of non-crossing spanning trees on $P$. Start each
tree with the same edge $T := p_1 p_2$. Repeatedly augment $T$ with one
edge at a time, such that the clockwise arc spanned by the latest edge increases by one,
either from its ``left'' (encode this choice by $0$), or from its ``right''
endpoint (encode this choice by $1$). Notice that regardless of the choice,
the latest edge has always the same length. As a result, the clockwise arc
spanning the new tree $T$ now includes one more point: either its right
endpoint moves clockwise by one position, or its left endpoint moves
counter-clockwise by one position. After $n-2$ steps, $T$ is a non-crossing
spanning tree of $P$. If $n=2k+1$, the sequence of edge lengths (including
the first edge, $p_1 p_2$) is
$$ \ell_1,\ldots,\ell_{k-1},\ell_k,\ell_{k-1},\ldots,\ell_1. $$
If $n=2k$, the sequence of edge lengths (including the first edge,
$p_1 p_2$) is
$$ \ell_1,\ldots,\ell_{k-1},\ell_k,\ell_k,\ell_{k-1},\ldots,\ell_1. $$
In both cases, the weight of each tree is $L(S_p)$.
All 0-1 sequences of length $n-2$ yield different trees, so there are at
least $2^{n-2}$ longest non-crossing spanning trees, as required.
This completes the proof of Theorem~\ref{T3}.
\qed

\subsection{Longest non-crossing tours}

We show here that the maximum number of longest non-crossing
spanning cycles is also exponential in $n$.
In the next section we also show that the maximum number of shortest
non-crossing spanning cycles on $n$ points is exponential in $n$
(Theorem~\ref{thm:min-tours}).

\begin{theorem} \label{thm:max-tours}
Let $\sc_{\rm max}(n)$ denote the maximum number of longest non-crossing
spanning cycles that an $n$-element point set can have.
Then we have
$ \Omega(2^{n/3}) \leq \sc_{\rm max}(n) \leq O(68.62^n) $.
\end{theorem}
\begin{proof}
Since we are counting non-crossing tours,
we have $\sc_{\rm max}(n) =O(68.62^n)$ by Theorem~\ref{thm:polygonize}.
It remains to show the lower bound. For every $k\in \mathbb{N}$, we
construct a set $Q$ of $4k+1$ points that admits $2^k=\Omega(2^{n/4})$
longest non-crossing tours. We start by constructing an auxiliary
set $P$ of $2k$ points. The auxiliary point set $P$ may contain
collinear triples, however our final set $Q$ does not.
Recall that two segments cross if and only if their relative interiors intersect.
We construct $P=\{c_i,x_i:i=1,2,\ldots , k\}$ with the following properties:
\begin{itemize}\itemsep -2pt
\item[{\rm (i)}] for every $x_i$, the farthest point in $P$ is $c_i$;
\item[{\rm (ii)}] the perfect matching $M=\{c_ix_i:i=1,2,\ldots , k\}$ is non-crossing; and
\item[{\rm (iii)}] the convex hull of $P$ is $\conv(P)=(x_1,c_1,c_2,\ldots , c_k)$.
\end{itemize}
Note that property (i) implies that $M$ is the maximum matching of $P$.

For $k\in \mathbb{N}$, let $\alpha=\frac{\pi}{3k}$.
We construct $P=\{c_i,x_i:i=1,2,\ldots , k\}$ iteratively. During the
iterative process, we maintain the properties that
$$|x_ic_i|>\max_{j<i}|x_ic_j|\hspace{1cm}\mbox{\rm and}\hspace{1cm}|x_{i+1}c_i|>\max_{j<i}|x_{i+1}c_j|.$$
Initially, let $c_1=(0,0)$, $x_1=(2,0)$,
and $x_2=(2-\frac{1}{k},0)$. Let $\vec{\ell}_1$ be a ray emitted by $x_1$ and incident
to $c_1$. Refer to Fig.~\ref{fig:longcycle}.
If $c_i$, $x_i$ and $x_{i+1}$ are already defined, we construct points
$c_{i+1}$ and $x_{i+2}$ (in the last iteration, only $c_{i+1}$)
as follows. Let $\vec{\ell}_{i+1}$ be a ray emitted
by $x_{i+1}$ such that $\angle(\vec{\ell}_{i+1},\vec{\ell}_i) = \alpha$.
Compute the intersections of ray $\vec{\ell}_{i+1}$ with the circle centered at $x_i$
of radius $|x_ic_i|$ and the circle centered at $x_{i+1}$ of radius
$|x_{i+1}c_i|$. Let $c_{i+1} \in \vec{\ell}_{i+1}$ be the midpoint of
the segment between these two intersection points.
This choice guarantees that
$|x_{i+1}c_{i+1}|>|x_{i+1}c_j|$ and $|x_jc_{i+1}|<|x_jc_j|$ for all $j\leq i$.
Now let $x_{i+2}\in c_{i+1}x_{i+1}$ be a point at distance at most
$\frac{1}{k}$ from $x_{i+1}$ such that we have
$|x_{i+2}c_{i+1}|>|x_{i+2}c_j|$ for all $j\leq i$.
This completes the description of $P$.

\begin{figure}[htbp]
\centerline{\epsfysize=1.5in \epsffile{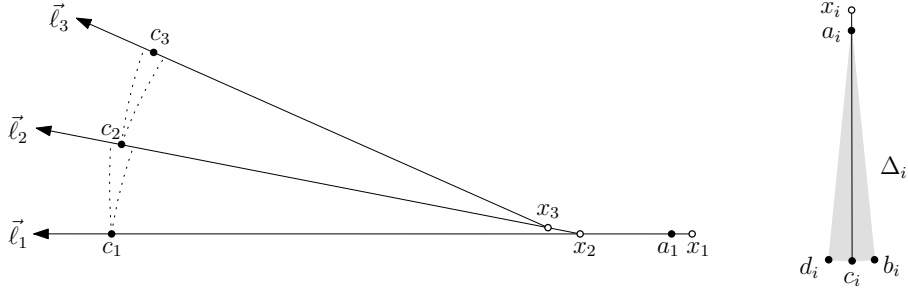}}
\caption{Left: The auxiliary point set $P$ for $k=3$.
Right: A long and skinny deltoid $\Delta_i=(a_i,b_i,c_i,d_i)$.}
\label{fig:longcycle}
\end{figure}

Note that $|x_ix_{i+1}|\leq \frac{1}{k}$ for $i=1,\ldots, k-1$, and so
the points $x_1,\ldots , x_k$ lie in a disk of diameter $1$.
Hence, for every point $x_i$, the farthest point in $P$ is
in $\{c_j:j=1,\ldots , k\}$. By the above construction, the farthest point from $x_i$
in $\{c_j:j=1,\ldots , k\}$ is $c_i$. This proves that $P$ has property (i). It is easy
to verify that $P$ has properties (ii) and (iii), as well.

We now construct the point set $Q$ based on $P$. Let $\delta>0$ be a sufficiently small
constant. For every segment $c_ix_i$ we construct a skinny deltoid
$\Delta_i=(a_i,b_i,c_i,d_i)$, see Fig.~\ref{fig:longcycle},
such that $a_i\in c_ix_i$ is at distance $\delta$ from $x_i$, we have
$|b_ic_i|=|c_id_i|=\delta$, and $|a_ib_i|=|a_ic_i|=|a_id_i|=|c_ix_i|-\delta$.
Since the segments $c_ix_i$ are pairwise non-crossing and $\delta>0$ is small,
the deltoids $\Delta_i$ are pairwise interior disjoint.
Let $Q$ be the set of vertices of all deltoids $\Delta_i$, $i=1,\ldots , k$, and
the point $x_1$. Since $\conv(P)=(c_1,c_2,\ldots , c_k,x_1)$, we have
$\conv(Q)=(b_1,c_1,d_1,b_2,c_2,d_2,\ldots , b_k,c_k,d_k,x_1)$, and
the points $\{a_i:i=1,\ldots , k\}$ lie in the interior of $\conv(Q)$.
If $\delta>0$ is sufficiently small, then the farthest points from $a_i$ in $Q$
are $b_i$, $c_i$, and $d_i$, for every $i=1,2,\ldots , k$.

Every non-crossing tour of $Q$ visits the convex hull vertices in the cyclic order
determined by $\conv(Q)$. We obtain a non-crossing tour by replacing
some edges of $\conv(Q)$ with non-crossing paths visiting the points
lying in the interior of $\conv(Q)$. If we replace either edge $b_ic_i$ or
$c_id_i$ with the path $(b_i,a_i,c_i)$ or $(c_i,a_i,d_i)$, respectively,
for every $i=1,2,\ldots , k$, then we obtain a tour. Let $\H$ be
the set of $2^k$ tours obtained in this way. These tours are
non-crossing, since for every $i$, we exchange an edge of
$\Delta_i$ with a path lying in $\Delta_i$, and the deltoids
$\Delta_i$ are interior disjoint. The tours in $\H$ have the
same length, $L=|\conv(Q)|-k\delta+2\sum_{i=1}^k|a_ic_i|$,
since $|a_ib_i|=|a_id_i|=|a_ic_i|$. It remains to show that this length
is maximal. Note that a non-crossing tour cannot have an edge between two
non-consecutive vertices of $\conv(Q)$. Hence, every edge that
intersects the interior of $\conv(Q)$ must be incident to some point
$a_i$ lying in the interior of $\conv(Q)$. Each $a_i$ is incident to two
edges of a tour: The total length of these two edges is at most $2|a_ic_i|$,
which is attained if $a_i$ is connected to $b_i$, $c_i$, or $d_i$.
In any other case, it is less than $2|a_ic_i|-k\delta$ if $\delta>0$ is
sufficiently small. The total length the edges on the boundary of the
convex hull is less than $|\conv(Q)|$. So any non-crossing tour in
which some point $a_i$ is not connected to $b_i, c_i$ or to $c_i,d_i$
must have length less than $L$. This implies $\sc_{\rm max}(n) = \Omega(2^{n/4})$.

To obtain the asserted bound, we use a skinny hexagon (instead of deltoid $\Delta_i$)
with five equidistant vertices on a circle centered at $a_i$.
We now have four possible ways to insert each $a_i$ into the tour, which implies
$\sc_{\rm max}(n) = \Omega(4^{n/6}) = \Omega(2^{n/3})$.
\end{proof}

Typically for the longest matching, spanning tree or spanning
cycle, one expects to see many crossings. Somewhat surprisingly, we
show that this is not always the case.

\begin{corollary} \label{cor:noncrossing}
For every even $n \geq 2$, there exists an $n$-element point set (in
general position) whose longest perfect matching is non-crossing.
\end{corollary}
\begin{proof}
For every $k \geq 1$, consider the set $\{ a_i, c_i : i=1,2,\ldots ,k \}$
of $2k$ points, and its perfect matching $\{a_ic_i: i=1,2,\ldots , k\}$
in the proof of Theorem~\ref{thm:max-tours}. This is the longest perfect
matching (over all perfect matchings of the point set, crossing or non-crossing),
and moreover, it is non-crossing.
\end{proof}

\section {Tours}\label{sec:tours}

In this section we derive estimates on the maximum multiplicity of
the shortest (minimum) and, respectively, the longest (maximum) Hamiltonian
tour on $n$ points (crossings allowed).
While the shortest tour has the non-crossing attribute for
free, the longest tour typically has crossings.

\subsection{Shortest tours}
\begin{theorem} \label{thm:min-tours}
Let $\sc_{\rm min}(n)$ denote the maximum number of shortest tours that
an $n$-element point set can have.
\begin{itemize}
\item [{\rm (i)}] For points in convex position, there is exactly
  one shortest spanning cycle, \ie $\sc_{\rm min}(n)=1$.
\item [{\rm (ii)}] For points in general position, $\sc_{\rm min}(n)$
is exponential in $n$. More precisely,
$ 2^{\lfloor n/3 \rfloor} \leq \sc_{\rm min}(n) \leq O(68.62^n) $.
\end{itemize}
\end{theorem}
\begin{proof}
\noindent {\rm \bf (i)}
It is well known (and easy to prove) that a shortest tour of a convex
point set does not admit any crossings, thus the vertices have to be
visited in clockwise or counterclockwise order.

\smallskip
\noindent {\rm \bf (ii)} Consider now a set $S$ of $n$ points in general
position. Since any minimum tour of $S$ is a non-crossing spanning cycle of $S$,
by Theorem~\ref{thm:polygonize} we have $\sc_{\rm min}(n) \leq \sc(n) = O(68.62^n)$.

A lower bound of $2^{\lfloor n/3 \rfloor}$ is given by the following construction.
Assume first that $n$ is a multiple of $3$. Consider an isosceles
triangle $\Delta{abc}$ with sides $\eps$, $\eps$ and $\eps/4$,
see Fig.~\ref{f2} (left).
\begin{figure}[htbp]
\centerline{\epsfysize=2.1in \epsffile{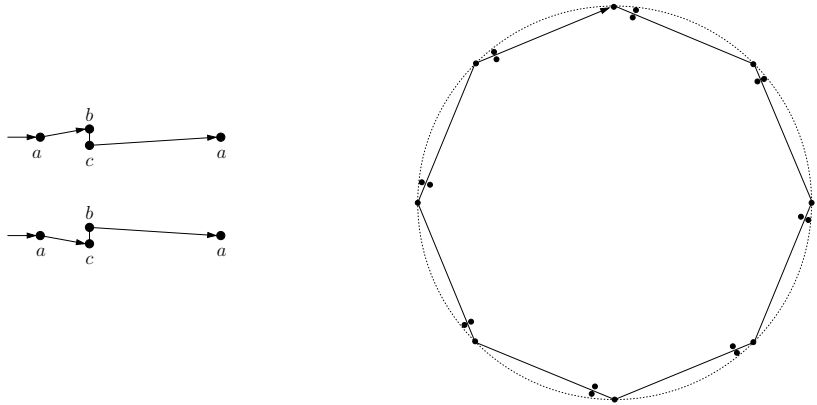}}
\caption{Lower bound for shortest tours.}
\label{f2}
\end{figure}
Let $S$ consist of $n/3$ rotated copies of $\Delta$
along a circle of unit radius, as in Fig.~\ref{f2} (right).
If $\eps$ is small enough, the shortest tour of $S$ must visit the three
vertices in each copy of $\Delta$ in their order along the circle.
We can assume w.l.o.g.\ that the groups (copies) are visited in clockwise order.
Then the first point visited in each group is (a rotated copy of) $a$.
Since each group of three points can be minimally traversed in two
ways, $a,b,c$, or $a,c,b$, the number of shortest tours is $2^{n/3}$.
The construction can be easily modified for any $n$,
by removing one or two points from one of the groups.
\end{proof}

\subsection{Longest tours for points in convex position}

In the following we give tight bounds for the number of maximum tours
on $n$ points in convex position (allowing crossings) and outline how to compute such
tours efficiently. The result for odd $n$ seems to have been first
discovered by Ito~\etal\cite{IUY01}. Nevertheless we present
our own simpler proof here for completeness.

\begin{theorem} \label{Thm-T}
Let ${\tt tc}_{\rm max}(n)$
denote the maximum number of longest tours that
an $n$-element point set in convex position can have.
For $n$ odd we have ${\tt tc}_{\rm max}(n)=1$ and the (unique) longest tour is a
thrackle. For $n$ even we have ${\tt tc}_{\rm max}(n)= n/2$.
\end{theorem}
\begin{proof}
Let $P$ be a set of $n$ points in convex position. Assume that $n=2k$
if $n$ is even, and $n=2k-1$ if $n$ is odd.
We pick an orientation for every possible tour arbitrarily.
 According to this orientation we denote an edge $xy$ of a tour as
 $x\ar y$ if $x$ is the
predecessor of $y$ in the tour. We call two edges $x\ar y$ and $u\ar
v$ \emph{parallel} if they are disjoint (including their endpoints)
and the segments $xu$ and $yv$ are disjoint. 
For instance, the edges $a_3\ar a_2$ and $b_2 \ar b_3$ in
Fig.~\ref{fig:eventour}(b) are parallel. Two disjoint edges that are not
parallel are called \emph{anti-parallel}.
We say two vertices of $P$ are \emph{consecutive} if they define a convex hull edge.
To prove the theorem we first prove three simple observations.

\medskip\noindent
{\it Observation (a): The longest tour does not contain
a pair of anti-parallel edges.}\\ 
Suppose the longest tour does contain anti-parallel edges $x\ar y$ and
$u \ar v$. We delete both edges and construct a new tour by adding
$x\ar u$ and $y\ar v$ and re-orienting the part of the old tour
between $u$ and $y$. The new tour is longer since in a convex
quadrilateral the sum of the diagonals exceeds the sum of two opposing
side lengths. 

\medskip\noindent
{\it Observation (b): In a longest tour with parallel edges
$x\ar y$ and $u\ar v$, the vertices $x,u$ and the vertices $y,v$ are
consecutive. In particular, a longest tour contains no triplet of
pairwise parallel edges.} \\ 
Assume that there exists a vertex $w$ ``in between'' $x$ and $u$
with predecessor $w'$. Then the edge $w'\ar w$ 
would determine a pair of anti-parallel edges with either 
$x\ar y$ or $u\ar v$, in contradiction to Observation (a).
By a similar argument one can show that $y$ and $v$ have to be consecutive.

\medskip\noindent
{\it Observation (c): The span of every edge in a longest
tour is at least $k-1$.} \\
Suppose, to the contrary, that the longest tour contains an edge
$x \ar y$ with span at most $k-2$. By the pigeonhole principle
there is at least one edge $u\ar v$ that does not cross $x \ar y$.
Due to Observation (b), $uv$ has to be an edge of span at least $n-k$.
Let $\{L,R\}$ be the partition of $P\setminus \{x,y\}$ induced by $x \ar y$
(that is, points left and right relative to $x \ar y$). Assume further that
$|R|>|L|$, that is, $\{u,v\}\subseteq R$. By Observation~(b),
a longest tour cannot contain three pairwise parallel edges. Hence,
in the longest tour the points in $R\setminus \{u\}$ have to have a
successor from $L\cup \{x\}$. But, again due to the pigeonhole principle,
this is impossible, since $|R|-1 \geq k-1$, but $|L|+1 \leq k-2$. 

\medskip
We now proceed with the proof of Theorem~\ref{Thm-T}. Assume first that $n=2k-1$ is odd. 
Due to observation (c), the longest tour does not have any edge of span at most $k-2$,
hence the span of every edge is $k-1$, which is the largest
possible span for $n$ odd. This determines the longest geometric tour.
For every vertex we have only two possible edges with span
$k-1$, hence the longest tour is unique.
\begin{figure}[htbp]
\begin{center}
\begin{tabular}{cp{15mm}c}
  \includegraphics[width=.28\columnwidth]{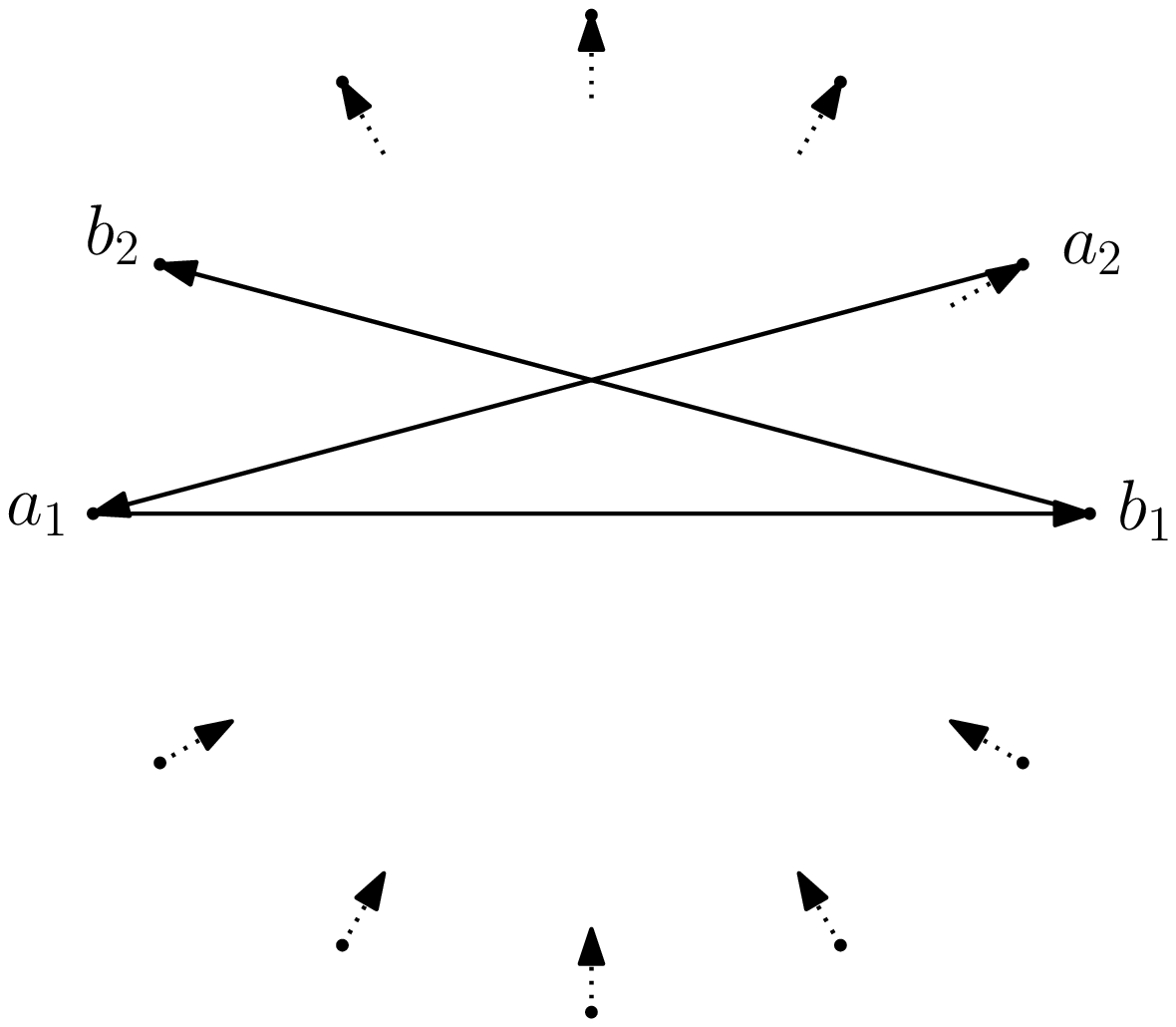} & &
  \includegraphics[width=.28\columnwidth]{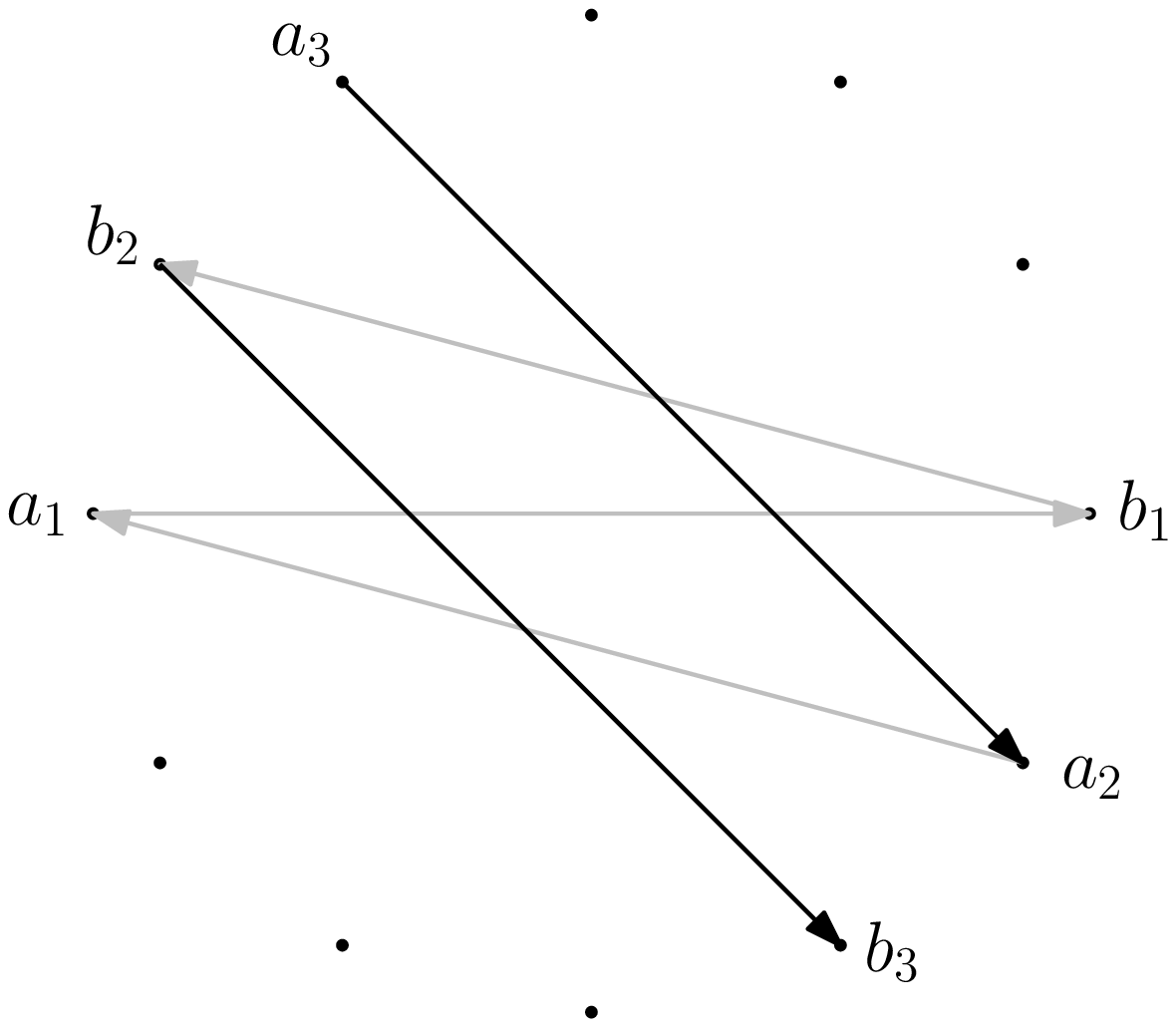} \\
  (a) && (b)
\end{tabular}
\caption{(a) A configuration that does not allow an extension to a
longest tour, since the possible 5 ``outgoing'' edges have only 4
candidates to connect with. (b) A valid extension. The edges
$a_{i+1}\ar a_i$ and $b_i\ar b_{i+1}$ form a pair of parallel
edges.}
\label{fig:eventour}
\end{center}
\end{figure}

Now assume that $n=2k$ is even. Notice that we have to use at least one
edge of span $k$, since otherwise (using edges with span $k-1$ only)
there would be two anti-parallel edges. Assume that a longest tour
contains the edge $a_1\ar b_1$ of span $k$. Let $a_2$ denote the predecessor of $a_1$,
and $b_2$ denote the successor of $b_1$ in the longest tour. The edges $a_2\ar a_1$ and
$b_1\ar b_2$ need to have span $k-1$.
There are two possibilities for the relative position of $a_2a_1$ and $b_2b_1$.
In case they cross, all outgoing edges from the vertices on the arc
between $a_1$ and $b_1$ (in cyclic order of $P$, opposite from $a_2$ and $b_2$)
have to cross $b_1\ar b_2$. This yields a contradiction by
the pigeonhole principle (see Fig.~\ref{fig:eventour}(a)).
Thus we can assume that $a_2$ and $b_2$ are located on opposite sides of
$a_1\ar b_1$, as in Fig.~\ref{fig:eventour}(b). We extend the partial tour step by step,
choosing new edges $a_{i+1}\ar a_i$, $b_i\ar b_{i+1}$ with span $k-1$
appropriately. If we would use an edge with span $k$ we would
``close'' the spanning cycle before all vertices have been visited.
We stop after we found $n-1$ edges and add the final
edge $b_{n/2}\ar a_{n/2}$ with span $k$. To see that this scheme works
correctly, notice that the edges added in every step form a pair of
parallel edges with span $k-1$, which ``rotates'' around $P$. The
construction can be characterized by the location of the edges with
span $k$. Since $a_1$ and $b_{n/2}$ are consecutive, as well as
$a_{n/2}$ and $b_{1}$, we have at most $n/2$ choices to select these
two edges.

If $P$ is the vertex set of a regular $n$-gon, then $P$ has
$n/2$ longest tours by rotational symmetry. The 5 longest tours for
$n=10$ are depicted in Fig.~\ref{fig:longtour10}.
\end{proof}

\begin{figure}[htbp]
\centerline{\epsfysize=1.2in \epsffile{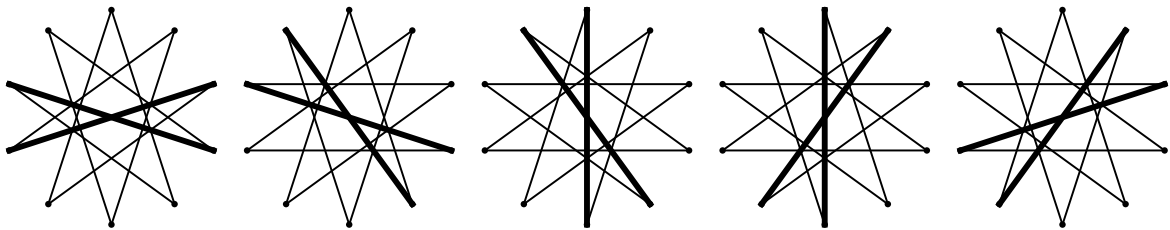}}
\caption{The 5 longest tours on the vertex set of a regular
  10-gon. The two edges with span $5$ are drawn with bold lines.}
\label{fig:longtour10}
\end{figure}

As an immediate algorithmic corollary of
Theorem~\ref{Thm-T}, the longest tours on a set of $n$
points in convex position can be computed in $O(n \log{n})$ time. 
If $n$ is odd and the convex polygon $P=p_0,p_1,\ldots,p_{n-1}$ is given, 
the tour can be computed in $O(n)$ time: start with $i=0$,
and iteratively set $i \leftarrow i+\frac{n-1}{2}$, where the indices
are taken modulo $n$.
If $n$ is even, we have $n/2$ candidates for the longest tour. 
Given a candidate, we can construct the next candidate by exchanging
only 4 edges of the tour. Thus, after computing the weight of the first candidate,
all other candidates can be evaluated in $O(1)$ per tour.
As in the case of odd $n$, we can compute the longest tours in
$O(n)$ time if the cyclic order of $P$ is given, and in $O(n\log n)$ time otherwise.

\section{Conclusion} \label{sec:conclusion}

Our investigations leave some questions open:

\smallskip
1. We used a special point configuration, $D(n,3^r)$, to prove a new lower
bound on $\tr(n)$, the maximum number of triangulations. The question arises if
this (or a similar) configuration could also give better lower bounds for other
quantities we studied, such as $\cf(n)$, $\st(n)$, or $\sc(n)$.
However, no analysis for other geometric graph classes has been done
so far, not even for the simpler-to-analyze \emph{double zig-zag chain} $D(n,1^r)$,.

\smallskip
2. While for the maximum number of triangulations, the upper and lower
bounds are reasonably close, this is not the case for non-crossing
spanning cycles. Recall that the current best lower bound, offered
by the double chain configuration is $\Omega(4.462^n)$, and our new upper bound
is only $O(68.62^n)$. Most likely, the lower bound is closer to the
truth, however the arguments are missing to justify this belief.
Clearly a large gap remains to be covered.

\smallskip
3. Is it possible to obtain upper bounds for the maximum
multiplicity of weighted configurations, better than those for the
corresponding unweighted structures? For example, is $\sc_{\rm max}(n)$
significantly smaller than $\sc(n)$, or $\pm_{\rm min}(n)$
significantly smaller than $\pm(n)$?

\smallskip
4. For points in convex position, we devised an $O(n\log n)$-time
algorithm for computing the longest tour.
However, for points in general position it is not known
whether this problem is NP-hard.
Barvinok~\etal\cite{BFJ+98} showed that the problem is solvable in
$O(n)$ time under the $L_1$ metric~\cite{BFJ+98}.
On the negative side, they also show that the problem is
NP-complete under the $L_2$ metric for points in 3-space.
Nothing is known about the complexity of computing the
longest non-crossing spanning tour, spanning path,
perfect matching or spanning tree. However, constant ratio
approximations are available for the last three problems;
see~\cite{ARS95,DT10}.


\begin{thebibliography}{99}

\bibitem{AHV+06} O. Aichholzer, T. Hackl, B. Vogtenhuber, C. Huemer,
F. Hurtado and H. Krasser,
On the number of plane geometric graphs,
\emph{Graphs and Combinatorics} {\bf 23(1)} (2007), 67--84.

\bibitem{ACNS82} M. Ajtai, V. Chv\'atal, M. Newborn and E. Szemer\'edi,
Crossing-free subgraphs,
\emph{Annals Discrete Math.} {\bf 12} (1982), 9--12.

\bibitem{ARS95} N. Alon, S. Rajagopalan, and S. Suri,
Long non-crossing configurations in the plane,
{\em Fundamenta Informaticae} {\bf 22} (1995), 385--394.
A preliminary version in \emph{Proc. 9th ACM Symp. on Comput. Geom.},
ACM Press, 1993, pp.~257--263.



\bibitem{BFJ+98}
A.~I.~Barvinok, S.~P.~Fekete, D.~S.~Johnson, A.~Tamir,
G.~J.~Woeginger, and R.~Woodroofe,
The geometric maximum traveling salesman problem,
\emph{J. ACM} {\bf 50(5)} (2003), 641--664.

\bibitem {BST98} P. Beame, M. Saks and J. Thathachar,
Time-space tradeoffs for branching programs,
{\it Proc. 39th Annual Symposium on Foundations of Computer Science},
1998, pp.~254--263.


\bibitem{BKK+07}
K. Buchin, C. Knauer, K. Kriegel, A. Schulz, and R. Seidel,
On the number of cycles in planar graphs,
\emph{Proc. 13th Annual International Conference on
  Computing and Combinatorics},
  volume 4598 of LNCS, Springer, 2007, pp.~97--107.


\bibitem{Dem10}
E.\ Demaine, Simple polygonizations,
\url{http://erikdemaine.org/polygonization/} (version of November, 2010).



\bibitem{D99} A. Dumitrescu,
On two lower bound constructions,
\emph{Proc.\ 11th Canadian Conf.\ on Computational Geometry}, 
1999, pp.~111--114.

\bibitem{D02} A. Dumitrescu,
On the maximum multiplicity of some extreme geometric configurations in the plane,
{\em Geombinatorics} {\bf 12(1)} (2002), 5--14.

\bibitem {DK01} A. Dumitrescu and R. Kaye,
Matching colored points in the plane: some new results,
{\em Computational Geometry: Theory and Applications}
{\bf 19(1)} (2001), 69--85.

\bibitem {DS00} A. Dumitrescu and W. Steiger,
On a matching problem in the plane,
{\em Discrete Mathematics} {\bf 211(1-3)} (2000), 183--195.

\bibitem{DT10} A. Dumitrescu and Cs. D. T\'oth,
Long non-crossing configurations in the plane,
{\em Discrete and Computational Geometry} {\bf 44(4)} (2010), 727--752.

\bibitem{DSST11} A. Dumitrescu, A. Schulz, A. Sheffer and Cs. D. T\'oth,
Bounds on the maximum multiplicity of some common geometric graphs,
\emph{Proc. 28th International Symposium on Theoretical Aspects of
  Computer Science}, Leibniz International Proceedings in Informatics
(LIPIcs) 9, 2011, pp.~637--648.


\bibitem{FN99} P. Flajolet and M. Noy,
Analytic combinatorics of non-crossing configurations,
{\em Discrete Mathematics} {\bf 204} (1999), 203--229.


\bibitem{GNT00} A. Garc\'{\i}a, M. Noy and A. Tejel,
Lower bounds on the number of crossing-free subgraphs of $K_N$,
{\em Computational Geometry: Theory and Applications}
{\bf 16(4)} (2000), 211--221.




\bibitem{HSSTW11}
M.~Hoffmann, M.~Sharir, A.~Sheffer, Cs.~D.~T\'oth, and E.~Welzl,
Counting plane graphs: flippability and its applications,
{\em Proc. Algorithms and Data Structure Symposium (WADS)},
vol.~6844 of LNCS, Springer, 2011, pp.~524--535.
Also at \verb+arxiv.org/abs/1012.0591+


\bibitem{HKRT08}
F.~Hurtado, M.~Kano, D.~Rappaport, and Cs.~D.~T\'oth:
Encompassing colored planar straight line graphs,
{\em Computational Geometry: Theory and Applications}
{\bf 39} (1) (2008), 14--23.

\bibitem{HN97} F. Hurtado and M. Noy,
Counting triangulations of almost-convex polygons,
{\em Ars Combinatoria} {\bf 45} (1997), 169--179.


\bibitem{IUY01} H. Ito, H. Uehara, and M. Yokoyama,
Lengths of tours and permutations on a vertex set of a convex polygon,
{\em Discrete Applied Mathematics} {\bf 115(1-3)} (2001), 63--71.




\bibitem{NM80}
M.~Newborn and W.~O.~J.~Moser,
Optimal crossing-free {H}amiltonian circuit drawings of $K_n$,
{\em J. Combin. Theory Ser. B} {\bf 29} (1980), 13--26.

\bibitem{PA95}
J. Pach and P.~K.~Agarwal,
{\em Combinatorial Geometry},
John Wiley, New York, 1995.


\bibitem{RSW08}
A.~Razen, J.~Snoeyink, and E.~Welzl,
Number of crossing-free geometric graphs vs. triangulations,
\emph{Electronic Notes in Discrete Mathematics}
{\bf 31} (2008), 195--200.

\bibitem{SS03} F. Santos and R. Seidel,
A better upper bound on the number of triangulations of a planar point set,
{\em Journal of Combinatorial Theory Ser.\ A} {\bf 102(1)} (2003),
186--193.


\bibitem{SW06} M. Sharir and E. Welzl,
On the number of crossing-free matchings, cycles, and partitions,
\emph{SIAM J.\ Comput.}
{\bf 36(3)} (2006), 695--720.

\bibitem{SW06b} M. Sharir and E. Welzl,
Random triangulations of point sets,
{\it Proc. 22nd Annual ACM-SIAM Symposium on Computational Geometry},
ACM Press, 2006, pp.~273--281.

\bibitem{SS10}
M. Sharir and A. Sheffer,
Counting triangulations of planar point sets,
{\em Electronic Journal of Combinatorics} {\bf 18} (2011), P70.



\end{thebibliography}
\end {document}